\documentclass[aps, pra, notitlepage, onecolumn,longbibliography,nofootinbib,10pt]{revtex4-2}

\usepackage{amssymb}
\usepackage{amsmath,amsthm}
\usepackage[colorlinks=true]{hyperref}
\usepackage{listings}
\usepackage{subfig}
\usepackage{algpseudocode}
\usepackage{algorithm}
\usepackage{placeins}

\usepackage[a4paper, margin=1in]{geometry}

\usepackage[utf8]{inputenc}
\usepackage{color}
\usepackage{graphicx}
\usepackage{empheq}
\usepackage[nameinlink]{cleveref}

\newcommand{\magn}[1]{\left|#1\right|}

\newcommand{\bra}[1]{\left\langle#1\right|}
\newcommand{\ket}[1]{\left|#1\right\rangle}
\newcommand{\bracket}[2]{\left\langle#1|#2\right\rangle}

\newcommand{\heis}[2]{\vec{\sigma}_{#1} \cdot \vec{\sigma}_{#2}}

\newcommand{\XEB}{\textsf{XEB}}

\newcommand{\Exp}{\mathop{\mathbb{E}}}
\newcommand{\Var}{\mathop{\text{Var}}}
\theoremstyle{definition}
\newtheorem{definition}{Definition}
\theoremstyle{plain}
\newtheorem{theorem}{Theorem}
\newtheorem{conjecture}{Conjecture}
\newtheorem{fact}{Fact}

\begin{document}

\title{On the complexity of sampling from shallow Brownian circuits}

\author{Gregory Bentsen$^1$}
\author{Bill Fefferman$^2$}
\author{Soumik Ghosh$^2$}
\author{Michael J. Gullans$^3$}
\author{Yinchen Liu$^{4,5}$}

\affiliation{$^1$Department of Physics, The College of William \& Mary}
\affiliation{$^2$Department of Computer Science, The University of Chicago}
\affiliation{$^3$Joint Center for Quantum Information and Computer Science,
University of Maryland and NIST}
\affiliation{$^4$Department of Combinatorics and Optimization, University of Waterloo}
\affiliation{$^5$ Institute for Quantum Computing, University of Waterloo}

\begin{abstract}
While many statistical properties of deep random quantum circuits can be deduced, often rigorously and other times heuristically, by an approximation to global Haar-random unitaries, the statistics of \textit{constant-depth} random quantum circuits are generally less well-understood due to a lack of amenable tools and techniques. We circumvent this barrier by considering a related constant-time Brownian circuit model which shares many similarities with constant-depth random quantum circuits but crucially allows for direct calculations of higher order moments of its output distribution. Using mean-field (large-$n$) techniques, we fully characterize the output distributions of Brownian circuits at shallow depths and show that they follow a Porter-Thomas distribution, just like in the case of deep circuits, but with a \emph{truncated} Hilbert space. The access to higher order moments allows for studying the expected and \textit{typical} Linear Cross-entropy (\textsf{XEB}) benchmark scores achieved by an ideal quantum computer versus the state-of-the-art classical spoofers for shallow Brownian circuits. We discover that for these circuits, while the quantum computer typically scores within a constant factor of the expected value, the classical spoofer suffers from an exponentially larger variance. Numerical evidence suggests that the \textit{same} phenomenon also occurs in constant-depth discrete random quantum circuits, like those defined over the all-to-all architecture. We conjecture that the same phenomenon is also true for random brickwork circuits in high enough spatial dimension.
\end{abstract}

\maketitle

\section{Introduction}
\label{sec:intro}

\noindent A major goal in the near-term quantum era is to achieve quantum
advantage — the first experimental demonstration of a quantum
computational speedup.  All current quantum advantage experiments
involve implementations of “random quantum circuits” in which several
layers of independently random quantum gates are applied to a simple
initial state followed by measurement of all qubits in the
computational basis \cite{Arute2019-zu, zhu2021quantumcomputationaladvantage60qubit, morvan2023phasetransitionrandomcircuit}.

Why are such random quantum circuits hard to simulate classically?
There are two perspectives on this question.  The traditional
perspective is that the hard problem is sampling from the output
distribution of the random quantum circuit.  There is rigorous
complexity-theoretic evidence that such sampling problems cannot be
solved by any efficient classical algorithm (see e.g.,
\cite{aaronson2016complexitytheoreticfoundationsquantumsupremacy, Bouland_2018, movassagh2020quantumsupremacyrandomcircuits, Bouland_2022}).  On the other hand, verifying
that the experimental output distribution is close to the ideal
quantum output distribution is a daunting computational task.  Current
experiments use tests such as the linear cross-entropy benchmark (or
\XEB), which assigns a numeric value to experimentally observed output
samples.  But such benchmarks only correspond to the fidelity of the
experimental output state under certain reasonable, but still
restrictive, noise assumptions
\cite{morvan2023phasetransitionrandomcircuit,ware2023sharpphasetransitionlinear}.

A second perspective is that the hard problem solved by random quantum
circuit sampling experiments is not necessarily to faithfully sample
from a hard distribution but merely to score sufficiently well on the
\XEB~benchmark \cite{aaronson2016complexitytheoreticfoundationsquantumsupremacy, aaronson2020classicalhardnessspoofinglinear}.  In a similar spirit
to a Bell inequality violation, the hope would be that scoring above
some threshold is a hard problem for any efficient classical
algorithm, and so can be used to demonstrate quantum advantage
directly.  If such a hope could be proven true, it would have the
advantage of not needing to rely on noise assumptions (i.e, since the
\XEB~score is something we can calculate directly with few experimental
samples).  On the other hand, we have comparatively less rigorous
evidence backing the classical hardness of scoring well on \XEB.
Indeed, recently developed classical algorithms have, in the case of
shallow depth random quantum circuits, been able to score well on \XEB~in expectation over the choice of circuits without sampling from the
output distribution of the ideal circuit \cite{barak2020spoofing, gao2024limitations}.
\subsection{Our contributions}
In this paper we revisit these recent algorithms for spoofing \XEB~with
a classical computer.  Our goal is to study the \emph{variance} of
these algorithms, to understand if the \XEB~scores of such algorithms
are typically as high as that scored by the ideal quantum computer, or
if this high expected score is merely because of the algorithm
achieving anomalously good scores on sporadic outlier circuits. \\ This
is a problem which had been left open due to the notorious difficulty
of computing higher order moments of output distributions of random
quantum circuits (i.e., computing the variance requires computing the
third and fourth order moments of such distributions which is challenging for known techniques). This is especially true for shallow circuits; even though there have been recent works studying random shallow circuits, e.g. see \cite{watts2024quantumadvantagemeasurementinducedentanglement, mcginley2024measurementinducedentanglementcomplexityrandom, NappShallow2022}, most properties of shallow random circuits are ill-understood.

In our work, we circumvent this barrier by studying a related model of
random quantum computation called Brownian circuits.  This model,
while not formally equivalent to the standard model of random quantum
circuits with Haar-random gates, has often been used as a close proxy
in the physics literature (e.g., \cite{Lashkari2013towards, bentsen2021measurement,
jian2023linear}). In the confines of this related model we prove that the
state-of-the-art \XEB~spoofing algorithm \cite{gao2024limitations} does indeed suffer
from a high variance, which gives evidence that these algorithms
perform similarly on random quantum circuits.  The same techniques
also allow us to conclude that this high variance does not occur in ideal Brownian circuits and to establish a lower bound on the typical
\XEB~scores obtained by these circuits. Numerical evidence suggests that the same phenomenon occurs in constant-depth discrete all-to-all random quantum circuits. We conjecture that it is also true for other architectures, like brickwork random circuits in high enough spatial dimension.

\section{Our method}
\noindent\textbf{Brownian circuits:} Computing variance requires access to the 3rd and 4th order moments of the random circuit distribution. To the best of our knowledge, there are no known techniques for bounding such moments directly in Haar-random unitary circuits. We circumvent this barrier by considering a related Brownian circuit model. An $n$-qubit Brownian circuit is defined by the unitary
\begin{equation}
    U = \prod_t U_t = \prod_t \exp{\left[-i \sum_{j < k}\sum_{\alpha,\beta\in\{x,y,z\}} J_{jk}^{\alpha \beta}(t) \sigma_j^{\alpha} \sigma_k^{\beta} dt \right]},
\end{equation}
where $\vec{\sigma}_j = (\sigma_j^{x},\sigma_j^{y},\sigma_j^{z})$ are the Pauli matrices acting on qubit $j$ and $dt$ is an infinitesimal timestep that we take to $dt \rightarrow 0$ at the end of the calculation. Here the couplings $J_{jk}^{\alpha \beta}(t)$ are Gaussian white-noise random variables with zero mean and variance
\begin{equation}
    \mathbb{E} \left[ J_{jk}^{\alpha \beta}(t) J_{j'k'}^{\alpha' \beta'}(t') \right] = \frac{J}{n dt} \delta_{jj'} \delta_{kk'} \delta^{\alpha \alpha'} \delta^{\beta \beta'} \delta_{tt'}
\end{equation}
where $J$ is a constant that sets the energy scale. Brownian circuits are continuous analogues of all-to-all random quantum circuit composed of Haar-random 2-qubit gates, and they have been used to model quantum information scrambling \cite{Lashkari2013towards}, entanglement transitions \cite{bentsen2021measurement}, and the convergence to $k$-designs \cite{jian2023linear}.
% In contrast, we are the first to study Brownian circuits in the constant-depth regime. \textcolor{red}{is this true?} 
The reason to study Brownian circuits, in place of other ensembles, is because it shares many similarities with discrete random quantum circuits, while also allowing for direct calculation of higher moments. More specifically, Brownian circuits are technically appealing due to the following reasons:
\begin{itemize}
\item For every $n$-qubit unitary $U$, we use $q_U(x)=|\bra{x}U\ket{0^n}|^2$ to denote its output distribution. For an algorithm $A$, quantum or classical, that on input a circuit $U$ outputs an $n$-bit string $x\in\{0,1\}^n$ drawn from some distribution $A_U$, the \textsf{XEB} score attained by the algorithm $A$ on the input $U$ is defined by
\begin{equation}
\textsf{XEB}(U,A)=2^n\sum_{x\in\{0,1\}^n}A_U(x)q_U(x)=2^n\Exp_{x\sim A_U}[q_U(x)].
\end{equation}

Whereas approaches are only known for bounding the first and second order moments of the output distribution of random quantum circuit, see for e.g. \cite{hunterjones2019unitarydesignsstatisticalmechanics, Dalzell_2022}, the Brownian circuit technology allows us to readily compute the $k$-th order moments of the distribution up to $k < D$ where $D=2^n$ is the dimension of the Hilbert space. 

%In principle this allows us to compute replica limits that can be used for calculating, for example, trace norms \cite{x}.

\item These circuits lead to an effective Hamiltonian that operates in continuous time, which is distinct from the discrete nature of random circuits. This continuous-time dynamics allows us to make direct contact with path integral methods that greatly simplify the analysis and to bring powerful tools from quantum field theory to bear on the problem. 

\item The all-to-all interactions in our model allow us to make use of large-$n$ (mean-field) methods in which the effective dynamics becomes semi-classical and readily solvable. These large-$n$ `saddle-point' methods are standard treatments for mean-field-theory models in the fields of classical spin glasses \cite{edwards1975theory,sherrington1975solvable,thouless1977solution,almeida1978stability,parisi1979toward,gross1984simplest,mezard1986spin,young1997spin}, quantum spin glasses \cite{bray1980replica,sachdev1993gapless,miller1993zerotemperature,read1995landau,kopec1995continuous,georges2000mean,georges2001quantum}, condensed matter physics \cite{parcollet1999nonfermi,fitzpatrick2014nonfermi,sachdev2015bekenstein,werman2017nonquasiparticle,davison2017thermoelectric,gu2020notes,chowdhury2022sachdev}, high-energy physics \cite{brezin1978planar,witten1980expansion,coleman1985aspects,thooft1993planar,moshe2003quantum,kitaev2015simple,maldacena2016remarks,fu2016numerical,gu2017local,berkooz2017higher,fu2017supersymmetric,kitaev2018soft,hartnoll2018holographic,sarosi2018ads2,saad2019semiclassicalrampsykgravity,rosenhaus2019introduction,berkooz2021complex}, and, more recently, quantum information science \cite{bentsen2021measurement,sahu2022entanglement,jian2023linear}.
\end{itemize} 
In the course of our calculations, we introduce new connections between tools from condensed matter and high-energy physics, like large-$n$ (mean-field) methods, and quantum information theory, which could be of independent interest. 
\vspace{6pt}

\noindent\textbf{Agreement with numerical simulations:} We perform numerical simulations to verify that the behaviors predicted by Brownian circuits are exhibited by discrete random unitary circuits in the all-to-all architecture of sufficiently high constant depth. This is consistent with the fact that we must assume $JT \gg 1$ in our Brownian circuit calculations (see Appendices for details). It indicates that just as in the case of deep circuits, the Brownian model is also a good approximation to the dynamics of shallow all-to-all random quantum circuits. \\

\noindent\textbf{On the complexity of shallow depth brickwork circuits:} Our work poses an interesting open question regarding whether the behavior of shallow depth random brickwork circuits, in a high enough spatial dimension, is asymptotically similar to random Brownian circuits. We conjecture that a similar behavior indeed holds; the rigorous statement can be found in \Cref{conj:D-dimensional}. If a connection can indeed be justified, then the toolkit we developed for Brownian circuits can be used to argue about the complexity of sampling from shallow depth random brickwork circuits.

\section{Our results in context}
\noindent We now describe our results in more detail.
\vspace{6pt}

\noindent\textbf{Typical \textsf{XEB} scores of the quantum computer vs. \cite{gao2024limitations} spoofer:} For every $n$-qubit unitary $U$, we use $q_U(x)=|\bra{x}U\ket{0^n}|^2$ to denote its output distribution. For an algorithm $A$, quantum or classical, that on input a circuit $U$ outputs an $n$-bit string $x\in\{0,1\}^n$ drawn from some distribution $A_U$, the \textsf{XEB} score attained by the algorithm $A$ on the input $U$ is defined by
\begin{equation}
\textsf{XEB}(U,A)=2^n\sum_{x\in\{0,1\}^n}A_U(x)q_U(x)=2^n\Exp_{x\sim A_U}[q_U(x)].
\end{equation}
We use $\textsf{XEB}(U,U)$ to denote the score obtained by a perfect quantum computer performing random circuit sampling. Let $\mathcal{B}$ denote the distribution of $n$-qubit time-$T$ Brownian circuits and let $\mathcal{D}$ denote the distribution of $n$-qubit depth-$d$ all-to-all random quantum circuits composed of Haar-random $2$-qubit gates. Our first calculations confirm that the expected \textsf{XEB} score achieved by the quantum computer for Brownian circuits is
\begin{equation}
\Exp_{U\sim\mathcal{B}}\left[\textsf{XEB}(U,U)\right]=2(1+e^{-24JT})^n,
\end{equation}
which has the same form of scaling w.r.t to $n$ and $T$ as a provable lower bound 
\begin{equation}
\Exp_{U\sim\mathcal{D}}\left[\textsf{XEB}(U,U)\right]\geq\left(1+\frac{1}{3}\left(\frac{2}{5}\right)^d\right)^n
\end{equation}
for $\mathcal{D}$ \cite{Dalzell_2022} w.r.t $n$ and $d$. Next, working out the 4th moment expression
\begin{equation}
\Exp_{U\sim\mathcal{B}}\left[\left(\Exp_{x\sim q_U}[q_U(x)]\right)^2\right]=\frac{4!}{2^{2n}} \left( 1 + e^{-24 J T} \right)^{2n}
\label{eq:quantum_4th_moment}
\end{equation}
allows us to conclude via the Paley-Zygmund inequality that with probability at least $\frac{1}{24}$ over $U\sim\mathcal{B}$, $\XEB(U,U)\geq \frac{1}{2}\Exp_{V\sim\mathcal{B}}[\XEB(V,V)]$, showing that the quantum computer typically achieves an \textsf{XEB} score within a constant factor of the expected value for shallow Brownian circuits. By computing another 3rd moment expression, we are also able to upper bound the number of string samples required to estimate $\textsf{XEB}(U,U)$ by an empirical average.

Now that we have understood the typical behavior of the quantum computer, we move on to analyze the spoofing algorithm of \cite{gao2024limitations}. The algorithm works by ``cutting up" the circuit horizontally into many small slices, deleting any gate that connects different slices, and then brute-force simulate each disjoint slice separately to sample from the resultant product distribution. While the slicing strategy originally stated in \cite{gao2024limitations} is tailored to geometrically-local circuits, we give a slight generalization and show that the general principle of \cite{gao2024limitations} does indeed translate to all-to-all circuits. Using $A$ to denote the spoofing algorithm, we prove, perhaps surprisingly given the simplicity of the algorithm, that
\begin{equation}
\Exp_{U\sim\mathcal{D}}[\textsf{XEB}(U,A)]\geq \left(1+\left(\frac{1}{15}\right)^d\right)^{\frac{n}{d^2+1}},
\end{equation}
which implies its scaling in $n$ is of the same form as that of the quantum computer for constant $d$. We show that (strictly speaking) another variant of the algorithm natural for Brownian circuits, again denoted by $A$, also achieves a similar scaling of
\begin{equation}
\Exp_{U\sim\mathcal{B}}[\textsf{XEB}(U,A)]=(1+e^{-24JT})^n.
\end{equation}
What about the variance of the spoofing algorithm for Brownian circuits? It turns out that
\begin{equation}\Exp_{U\sim\mathcal{B}}\left[\left(\Exp_{x\sim A_U}[q_U(x)]\right)^2\right]=\frac{2^{K+1}}{2^{2n}}(1+e^{-24JT})^{2n}
\label{eq:harvard_4th_moment}
\end{equation}
where $K$ is the number of disjoint slices created by the algorithm. Since $K$ needs to be at least $\Omega(n/\log n)$, by comparing \Cref{eq:quantum_4th_moment} and \Cref{eq:harvard_4th_moment}, we see that the spoofing algorithm suffers from an exponentially larger variance compared to the quantum computer for shallow Brownian circuits. By numerically estimating these 4th moment quantities w.r.t $\mathcal{D}$, we observe the same effect in shallow all-to-all random quantum circuits, which is evidence in favor of a similar conclusion.
\vspace{6pt}

\noindent\textbf{Local Porter-Thomas behavior:} It is known that at \emph{large depths}, both the Brownian circuit ensemble, as well as the all-to-all connected random circuit ensemble, converge to the Porter-Thomas distribution. However, it is less clear what happens at \emph{shallow} depths. As our second result, we fully characterize the output distribution of Brownian circuits at shallow depths and show that the return probability of observing all zeros in the output, with the randomness over the choice of the circuit, follows a Porter-Thomas distribution, just like in the case of deep circuits, but with a \emph{truncated} Hilbert space. The extent of the truncation depends on the depth of the circuit---shallower the depth, more is the truncation. More concretely, we show that for the Brownian circuit ensemble $\mathcal{B}$, the probability density function of the return probability $p=p(0^n)$ is
\begin{equation}
    Q(p) = b e^{-b p},
\end{equation}
where $b$, the effective Hilbert space dimension, depends on the time as
\begin{equation}
    b = 2^n (1+e^{-12 J T})^{-n}
\end{equation}
for Brownian circuit evolution time $T$, where the evolution time $T$ plays the role of circuit depth.

Note that the output distribution of a deep random quantum circuit exhibits a Porter-Thomas distribution, as verified in \cite{Arute2019-zu}. There are signatures of anti-concentration---meaning, the support encompasses almost all strings and they each have a substantial probability mass---at logarithmic depth itself \cite{Dalzell_2022, Deshpande_2022}. This is also borne out in the Brownian case by our analysis, which indicates deep similarities between the two models and indicates that the output distribution of shallow random quantum circuits could also follow a truncated Porter-Thomas distribution.

%Our results posit two interesting perspectives on the complexity of random circuit sampling at shallow depths. If the connections between Brownian circuits and discrete unitary circuits can be made more formal, these results strongly suggest verifiable hardness of sampling at shallow depths for discrete unitary circuits. If there is a discrepancy between the two behaviors, it raises interesting questions about the nature of differences between these two models and raises the need to find better ways to compute, or reason about, higher moments of shallow random circuits. \\

\section{Background}
\label{sec:setup}
% At a bare minimum, quantum advantage seeks to identify a computational problem that is strictly ``easier'' to solve using a (perfect) quantum computer than classical computers. For this express purpose, it is therefore natural to consider a task which is by definition quantumly ``easy'', seemingly difficult to design classical algorithms for, and try to prove its classical hardness. Random circuit sampling (RCS) is a leading candidate born out of this approach. 

A random quantum circuit is a circuit drawn randomly according to some random circuit distribution. A random circuit distribution is defined by a circuit architecture, which itself could involve randomness, specifying the placement of each $2$-qubit gate in the circuit and a probability distribution over $\mathbb{U}(4)$ to draw each $2$-qubit gate independently at random from. In this paper, we will only consider random quantum circuits composed of gates drawn from the Haar measure, which is the continuous uniform distribution over $\mathbb{U}(4)$. As such, it is unambiguous to refer to a random circuit distribution by the circuit architecture that induces it. This paper mainly focuses on the all-to-all circuit architecture which we define as follows.

\begin{definition}[All-to-all Random Quantum Circuit Architecture]
For every even integer $n\geq 2$ and integer $d\geq 1$, the $n$-qubit depth-$d$ all-to-all random quantum circuit architecture is denoted by $\mathcal{D}_{n,d}$. In this architecture, for every layer of gates $i\in\{1,\ldots,d\}$, a uniformly random permutation $\sigma_i:\{1,\ldots,n\}\rightarrow\{1,\ldots,n\}$ is chosen, and for ever $j\in\{1,\ldots,\frac{n}{2}\}$, there is a random $2$-qubit gate acting on qubits $\sigma_i(2j-1)$ and $\sigma_i(2j)$.
\end{definition}

\noindent Our definition ensures that every qubit will be acted on by at least one gate even at $d=1$ and that each layer consists of exactly $\frac{n}{2}$ $2$-qubit gates.

We are now ready to define the task of random circuit sampling.

\begin{definition}[Random Circuit Sampling (\textsf{RCS})]
On input an integer $c=O(1)$ and an $n$-qubit random quantum circuit $U$ drawn from some $n$-qubit random circuit distribution, the task is to output a bit string sampled from some distribution $P$ over $\{0,1\}^n$ such that
\begin{equation}
\sum_{x\in\{0,1\}^n}\left|P(x)-\left|\bra{x}U\ket{0^n}\right|^2\right|\leq\frac{1}{n^c}.
\label{eq:RCS}
\end{equation}
\end{definition}

\noindent Note that besides the need to truncate each entry of a gate to finite precision, \textsf{RCS} constitutes a proper classical-input-classical-output computational task. The $\frac{1}{n^c}$ error tolerance is not meant for tolerating noise in a quantum computer but rather to ensure fairness for classical randomized algorithms because for example, in a model of computation where a classical computer can only flip fair coins, it is impossible for the computer to flip a $(\frac{1}{3},\frac{2}{3})$-biased coin exactly. Consequently, it may be impossible for a classical computer to sample exactly from the output distribution of $U$, hence some approximation errors are inevitable. Nevertheless, \Cref{eq:RCS} does correspond to the accuracy achievable by fault-tolerant quantum computers.

It is evident that \textsf{RCS} can be solved efficiently using a (perfect, all-to-all connected) quantum computer by just running the circuit.

\begin{definition}[The Quantum Simulation Algorithm]
\label{def:QS}
On input an $n$-qubit random quantum circuit $U$, prepare the state $U\ket{0^n}$ by running the circuit $U$ on input $\ket{0^n}$. Then measure all qubits in the computational basis to obtain a bit string $x\in\{0,1\}^n$ and output $x$.
\end{definition}

Clearly, the above algorithm samples from the output distribution of $U$ given by $q_U(x)=|\bra{x}U\ket{0^n}|^2$ for every $x\in\{0,1\}^n$ and satisfies \Cref{eq:RCS}. It may even seem more intuitive to take \Cref{def:QS} as the definition of \textsf{RCS}, but the perspective of viewing ``doing \textsf{RCS} on a quantum computer'' (i.e. \Cref{def:QS}) as an algorithm will be important to our discussions.

There is a long line of work \cite{aaronson2016complexitytheoreticfoundationsquantumsupremacy, Bouland_2018, movassagh2020quantumsupremacyrandomcircuits, Bouland_2022} towards proving the classical hardness of \textsf{RCS} for which the ultimate goal can be encapsulated in the following conjecture.
\begin{conjecture}[\textsf{RCS} Supremacy Conjecture]
\label{conj:RCS}
There exists an integer $a\geq 1$ and (a uniform family of) random circuit distributions $\{\mathcal{C}_n\}_{n\geq 2}$ such that if there exists a randomized classical algorithm that can solve \textsf{RCS} with success probability at least $1-\frac{1}{n^a}$ over the choice of the circuit $U\sim\mathcal{C}_n$ in time polynomial in $n$, then the polynomial hierarchy collapses. 
\end{conjecture}

As stated, \Cref{conj:RCS} is yet to be proven; see \cite{Bouland_2022, krovi2022average} for the current frontier. As far as we know, it is consistent with all known results that \Cref{conj:RCS} holds for the all-to-all architecture $\mathcal{D}_{n,d}$ if the circuit depth $d$ is at least some sufficiently large constant \cite{NappShallow2022}. While efforts toward proving \Cref{conj:RCS} have faced obstacles, it could help our understanding by identifying conjectures stronger than \Cref{conj:RCS} that are not obviously false \cite{Aaronson2017supremacy}. A few of these conjectures are based on a quantity known as the linear cross-entropy benchmark.

\begin{definition}[Linear Cross-entropy Benchmark (\XEB)]
For an algorithm $A$ that on input a quantum circuit $U$ outputs an $n$-bit string $x\in\{0,1\}^n$ drawn from some distribution $A_U$, the \textsf{XEB} score attained by the algorithm $A$ on the input $U$ is defined by
\begin{equation}
\textsf{XEB}(U,A)=2^n\sum_{x\in\{0,1\}^n}A_U(x)q_U(x)=2^n\Exp_{x\sim A_U}[q_U(x)].
\end{equation}
We use $\textsf{XEB}(U,U)$ to denote the score obtained using \Cref{def:QS}. The expected \textsf{XEB} score attained by the algorithm $A$ w.r.t an $n$-qubit random circuit distribution $\mathcal{C}$ is
\begin{equation}
\Exp_{U\sim\mathcal{C}}[\textsf{XEB}(U,A)].
\end{equation}
\end{definition}

One of the main tasks concerning \textsf{XEB} is to design algorithms, quantum or classical, that achieve high expected \textsf{XEB} scores. Obviously, the largest possible score is achieved by the algorithm that finds and outputs a string $x^*$ that maximizes $\left|\bra{x}U\ket{0^n}\right|^2$ over all $x\in\{0,1\}^n$. However, for our purpose, it is only meaningful to consider expected \textsf{XEB} w.r.t random circuit distributions for which \Cref{conj:RCS} is believed to hold, so it is unlikely that this ``optimal'' strategy can be carried out efficiently for those distributions. On the opposite extreme, the trivial algorithm that outputs a uniformly random $n$-bit string receives a score of $1$ regardless of what the random circuit distribution is. For the quantum computer, \textsf{RCS} in the sense of \Cref{def:QS} lies in between the two extremes. Starting from the somewhat deep circuit depth regime of $d\geq\Omega(\log n)$, it is known that \Cref{def:QS} achieves a score approaching $2$ from above as $d$ increases for the all-to-all architecture \cite{Dalzell_2022}. Furthermore, it is conjectured that there is a gap between the performance of \Cref{def:QS} and the best possible classical algorithm.

\begin{conjecture}[Classical Hardness of \textsf{XEB} Conjecture \cite{Arute2019-zu}]
\label{conj:google}
No efficient classical algorithm can achieve an expected \textsf{XEB} score of $1+\varepsilon$ for any constant $\varepsilon$ for deep random quantum circuits.
\end{conjecture}

\noindent In other words, \Cref{conj:google} asserts that in the deep circuit depth regime, the trivial algorithm that samples uniformly at random from $\{0,1\}^n$ is essentially the optimal efficient classical algorithm. Note that \Cref{conj:google} is stronger than \Cref{conj:RCS} since if \Cref{conj:RCS} does not hold for deep random quantum circuits, then \Cref{conj:google} is falsified by the classical algorithm that simulates \textsf{RCS}.

In this work, we study the utility of \textsf{XEB} in the shallow depth regime where $d$ is some sufficiently large constant. For constant $d$, it is known that \Cref{def:QS} achieves a score that is exponentially increasing in $n$.

\begin{theorem}[Expected \textsf{XEB} Score of \Cref{def:QS} for the all-to-all architecture \cite{Dalzell_2022}]
For every even $n\geq 2$ and every $d\geq 1$, it holds that
\begin{equation}
\Exp_{U\sim\mathcal{D}_{n,d}}[\textsf{XEB}(U,U)]\geq 2^{\frac{1}{3}\left(\frac{1}{5}\right)^dn}.
\end{equation}
\label{thm:QS_xeb}
\end{theorem}

Non-trivial classical spoofing algorithms are also known in the low-depth regime. These are efficient classical algorithms that provably achieve expected \textsf{XEB} scores higher than the trivial algorithm. In this paper, we focus on the algorithm proposed by \cite{gao2024limitations} whose general strategy proceeds as follows.

\begin{definition}[The Spoofing Algorithm of \cite{gao2024limitations}]
On input an $n$-qubit random quantum circuit $U$, partition the qubits $\{1,\ldots,n\}$ into $K$ disjoint subsets, each having size equal to some $m_i\leq O(\log n)$. Modify the input $U$ to obtain a circuit $\tilde{U}=\tilde{U}_1\otimes\cdots\otimes \tilde{U}_K$ as follows. For each gate $V$ acting on qubits $j$ and $k$ in $U$, if $j$ and $k$ belong to different subsets constructed earlier, then replace the gate $V$ by the $2$-qubit identity gate. For each $i\in\{1,\ldots,K\}$, compute the state vector $\tilde{U}_i\ket{0^{m_i}}$ and use it to sample $z_i\in\{0,1\}^{m_i}$ with probability $|\bra{z_i}\tilde{U}_i\ket{0^{m_i}}|^2$. Output $x=z_1\cdots z_K$.
\label{def:HA}
\end{definition}

Generally speaking, the partition of the qubits needs to depend on the circuit layout specified by the architecture, so for deterministic (e.g. geometrically-local) architectures, the same partition can be used regardless of the input circuit instance. However, since the circuit layout could be different for each random circuit instance drawn from the all-to-all architecture, the partition needs to depend on the input circuit.

In \Cref{app:harvard_discrete}, we describe a simple greedy partition strategy with which we obtain a fully specified algorithm. While the efficiency of \Cref{def:HA} is evident, it is perhaps surprising that similar to \textsf{RCS}, the expected \textsf{XEB} score achieved by \Cref{def:HA} also increases exponentially in $n$ in the constant-depth regime.

\begin{theorem}
Let $A$ denote the spoofing algorithm of \cite{gao2024limitations} using the greedy partition strategy described in \Cref{alg:greedy_partition}. For every even $n\geq 2$ and $d\geq 1$, it holds that
\begin{equation}
\Exp_{U\sim\mathcal{D}_{n,d}}[\textsf{XEB}(U,A)]\geq \left(1+\left(\frac{1}{15}\right)^d\right)^{\frac{n}{d^2+1}}.
\end{equation}
\label{thm:harvard_expected_lower_bound}
\end{theorem}

\noindent Our proof of \Cref{thm:harvard_expected_lower_bound} can be found in \Cref{app:harvard_discrete}. Potentially, a tighter analysis or a better partitioning strategy may lead to a better scaling w.r.t $d$. Already, the scaling w.r.t $n$ is of the same form as that of the quantum computer established by \Cref{thm:QS_xeb}, both are exponentially increasing $n$ for constant $d$. Therefore, this leads us to conclude that expected \textsf{XEB} is asymptotically easy to spoof classically for $\mathcal{D}_{n,d}$ for constant $d$. 

Using the close connection between deep random quantum circuits and Haar-random unitaries, some typicality results are known in the deep circuit depth regime \cite{aaronson2016complexitytheoreticfoundationsquantumsupremacy,barak2020spoofing, kretschmer2021quantum}. For example, it is well-known that a perfect quantum computer achieves a score of $2$ with high probability over the choice of the deep circuit. On the contrary, all known results related to \textsf{XEB} for shallow random quantum circuits are limited to expected \textsf{XEB}. For instance, it is unknown whether the quantum computer gets a score close to the expectation with high probability over the choice of $U\sim\mathcal{D}_{n,d}$ for constant $d$. This also leaves open the possibility that, for example, the algorithm of \cite{gao2024limitations} could achieve a very low score, potentially worse than the score of $1$ guaranteed by the trivial algorithm, on most random input instances.

\section{Our results for Brownian circuits}

Here we present our primary results for the moments of the output distribution for Brownian circuits in the large-$n$ limit. We also present results for the moments of output distribution obtained from the spoofing algorithm of \cite{gao2024limitations}. We refer the reader to the Appendices for complete derivations of these results.

First, the moments of the output distribution $q_U(x)=|\bra{x}U\ket{0^n}|^2$ for Brownian circuits in the large-$n$ limit are:
\begin{equation}
    \Exp_{U\sim\mathcal{B}}[q_U(x)^k] = \frac{k!}{2^{kn}} \left(1 - e^{-12 J T} \right)^{k \magn{\mathbf{x}}} \left(1 + e^{-12 J T} \right)^{k (n-\magn{\mathbf{x}})} 
    \label{eq:brownianmoments}
\end{equation}
for any fixed output bitstring $x$ and for $k < 2^n$, where $\magn{x}$ is the Hamming weight. Performing a sum over the bitstrings $x$, we obtain
\begin{equation}
    \sum_{x\in\{0,1\}^n} \Exp_{U\sim\mathcal{B}}[q_U(x)^k] = \frac{k!}{2^{kn}} \left[ \left(1 - e^{-12 J T} \right)^{k} +  \left(1 + e^{-12 J T} \right)^{k} \right]^n
\end{equation}
We can also compute moments of the form
\begin{equation}
\Exp_{U\sim\mathcal{B}}\left[q_U(x)^{c}q_U(y)^{c}\right]= \frac{(2c)!}{2^{2cn}} \left( 1- e^{-12 J T} \right)^{(\magn{x}+\magn{y})c} \left( 1 + e^{-12 J T} \right)^{(2n - \magn{x} -\magn{y})c}
\end{equation}
for $2c = k$ even, which generalizes Eq. \eqref{eq:brownianmoments}. Summing over the bitstrings $x,y$ we obtain
\begin{equation}
\sum_{x,y\in\{0,1\}^n} \Exp_{U\sim\mathcal{B}}\left[q_U(x)^{c}q_U(y)^{c}\right] = \frac{(2c)!}{2^{2cn}} \left[ \left( 1 - e^{-12 J T} \right)^{c} + \left( 1 + e^{12 J T} \right)^{c} \right]^{2n}
\end{equation}

Second, we can compute moments of the spoofing algorithm of \cite{gao2024limitations}
\begin{equation}
\Exp_{U\sim\mathcal{B}}\left[q_U(x)^{c}A_U(x)^{c}\right] = \frac{(c!)^{1+K}}{2^{2cn}} \left[ \left(1 - e^{-12 J T} \right)^{2c\magn{x}} \left(1 + e^{-12 J T} \right)^{2c(n-\magn{x})} \right]
\label{eq:harvardbrownian}
\end{equation}
for $2c = k$ even. Summing over bitstrings $x$ we obtain:
\begin{equation}
    \sum_{x\in\{0,1\}^n} \Exp_{U\sim\mathcal{B}} \left[q_U(x)^{c}A_U(x)^{c}\right] = \frac{(c!)^{1+K}}{2^{2cn}} \left[ \left(1 - e^{-12 J T} \right)^{2c} + \left(1 + e^{-12 J T} \right)^{2c} \right]^n
\end{equation}
These expressions constitute the main technical results that we can achieve using Brownian circuit technology. See the Appendices for complete derivations and explanation of how to apply the algorithm of \cite{gao2024limitations} to all-to-all Brownian circuits.

\section{Techniques}
Our results rely on mean-field (large-$n$) techniques that facilitate the calculation of higher moments of the random circuit distribution. Two simplifications occur in our calculations that render the resulting dynamics analytically tractable. First, we are able to map our Brownian circuit dynamics onto imaginary time evolution under an effective all-to-all Hamiltonian. Second, because this Hamiltonian is all-to-all, we may analyze its spectrum using mean field theory. This is most easily accomplished by transforming the imaginary-time dynamics into a path integral, where large-$n$ methods are readily applied. Together, these two simplifications allow us to calculate the moments of the circuit distribution by saddle point, which is effectively a semiclassical calculation where $1/n$ plays the role of $\hbar$.

To see these techniques in action, let us consider the first moment of the return probability
\begin{equation}
    \mathbb{E}[p(0)] = \mathbb{E}_{U \sim \mathcal{E}} \magn{\bra{0} U \ket{0}}^2
\end{equation}
where we may write the return probability as an amplitude on a doubled Hilbert space:
\begin{equation*}
    p(0) = \magn{\bra{0} U \ket{0}}^2 = \bra{0} U \ket{0} \bra{0} U^* \ket{0} = \bra{00} U \otimes U^* \ket{00}.
\end{equation*}
After some relatively simple manipulation, we are able to write this return probability as a thermal expectation value
\begin{equation*}
    \mathbb{E}[p(0)] = \mathrm{Tr} \left[ e^{- H_{\mathrm{eff}} T} \mathcal{O} \right]
\end{equation*}
where the operator $\mathcal{O}$ is related to the initial $\ket{00}$ and final states $\bra{00}$ (see Appendices for details). This thermal expectation value is governed by the effective Hamiltonian:
\begin{equation}
    H_{\mathrm{eff}} = - \frac{J}{2 n} \left( \sum_i \vec{\sigma}_{iL} \cdot \vec{\sigma}_{iR} \right)^2 + \frac{9 J}{2} n + \mathcal{O}(1)
\end{equation}
where $L,R$ denote the `forward' $U$ and `backward' $U^*$ Hilbert spaces, and the $\mathcal{O}(1)$ refers to terms that are subleading in $n$.
Here the mean-field nature of the Hamiltonian is evident, as each operator $\vec{\sigma}_{jL} \cdot \vec{\sigma}_{jR}$ on site $j$ couples only to the mean field
\begin{equation}
    G_{LR} = \frac{1}{n} \sum_i \vec{\sigma}_{iL} \cdot \vec{\sigma}_{iR}.
    \label{eq:constraint}
\end{equation}
Because $G_{LR}$ consists of a large number $n$ of terms, the fluctuations of this operator are suppressed by powers of $1/n$. This is similar to the suppression of fluctuations that occurs in the central limit theorem. As a result, the field $G_{LR}$ may be treated semiclassically, and the dynamics, in the large-$n$ limit, is governed by the effective Hamiltonian
\begin{equation}
    H_{\mathrm{eff}} = -J \sum_i \left( \heis{iL}{iR} \right) G_{LR} + \frac{J}{2} G_{LR}^2 n + \frac{9 J}{2} n
\end{equation}
subject to the self-consistency constraint Eq.~\eqref{eq:constraint}. This Hamiltonian is separable, and therefore readily solved. Similar techniques generalize to higher moments, as well as moments of the output distribution for the spoofing algorithm of \cite{gao2024limitations}. We leave the details of these calculations to the Appendices.

\section{Typical \textsf{XEB} score obtained by a quantum computer for Brownian circuits}
In this section, we give a bound on the typical \textsf{XEB} scores achieved by a quantum computer for Brownian circuits with an upper bound on the number of string samples required for estimating the \textsf{XEB} score of a typical circuit instance. Because of our direct command of the 3rd and 4th order moments of Brownian circuits, the rest of the calculations boil down to standard applications of concentration inequalities. Recall that for every $n$-qubit unitary $U$, we use $q_U(x)=|\bra{x}U\ket{0^n}|^2$ denote the output distribution of $U$ for $x\in\{0,1\}^n$, and we use $\mathcal{B}$ to denote the distribution of $n$-qubit time-$T$ Brownian circuits.

To derive our bounds, we recall that we have previously computed the expected \textsf{XEB} score itself
\begin{equation}
\Exp_{U\sim\mathcal{B}}\left[\Exp_{x\sim q_U}[q_U(x)]\right]={\sum_{x\in\{0,1\}^n}\Exp_{U\sim\mathcal{B}}[q_U(x)^2]}=\frac{2}{2^{n}} \left( 1 + e^{-24 J T} \right)^n,
\label{eq:2nd_moment}
\end{equation}
the 3rd moment expression 
\begin{equation}
\Exp_{U\sim\mathcal{B}}\left[\Exp_{x\sim q_U}[q_U(x)^2]\right]=\sum_{x\in\{0,1\}^n}\Exp_{U\sim\mathcal{B}}[q_U(x)^3]= \frac{3!}{2^{2n}} \left( 1 + 3 e^{-24 JT} \right)^n,
\label{eq:3rd_moment}
\end{equation}
and the 4th moment expression
\begin{equation}
\Exp_{U\sim\mathcal{B}}\left[\left(\Exp_{x\sim q_U}[q_U(x)]\right)^2\right]={\sum_{x,y\in\{0,1\}^n}\Exp_{U\sim\mathcal{B}}\left[q_U(x)^2q_U(y)^2\right]}= \frac{4!}{2^{2n}} \left( 1 + e^{-24 J T} \right)^{2n}.
\label{eq:4th_moment}
\end{equation}

Using \Cref{eq:2nd_moment}, \Cref{eq:4th_moment}, and the Paley-Zygmund inequality, we get
\begin{equation}
\Pr_{U\sim\mathcal{B}}\left(\Exp_{x\sim q_U}[q_U(x)]\geq\frac{1}{2}\Exp_{V\sim\mathcal{B}}\left[\Exp_{x\sim q_V}[q_V(x)]\right]\right)\geq\frac{\left(\Exp_U[\Exp_{x\sim q_U}[q_U(x)]]\right)^2}{4\Exp_{U}\left[(\Exp_{x\sim q_U}[q_U(x)])^2\right]}=\frac{1}{24}.
\end{equation}
By Markov's inequality,
\begin{equation}
\Pr_{U\sim\mathcal{B}}\left(\Exp_{x\sim q_U}[q_U(x)^2]\geq 600\Exp_V\left[\Exp_{x\sim q_V}[q_V(x)^2]\right]\right)\leq\frac{1}{600}.
\end{equation}
Therefore, by the union bound, with probability at least $\frac{1}{25}$, we draw a Brownian circuit instance $U\sim\mathcal{B}$ satisfying
\begin{equation}
\Exp_{x\sim q_U}[q_U(x)]\geq\frac{1}{2}\Exp_{V\sim\mathcal{B}}\left[\Exp_{x\sim q_V}[q_V(x)]\right]
\end{equation}
and
\begin{equation}
\Exp_{x\sim q_U}[q_U(x)^2]\leq 600\Exp_{V\sim\mathcal{B}}\left[\Exp_{x\sim q_V}[q_V(x)^2]\right].
\end{equation}
Now suppose we have drawn such a $U$. Let $m\geq 1$ and suppose we draw $m$ independent samples $x_1,\ldots,x_m$ from the output distribution $q_U$. Consider the sample mean
\begin{equation}
\tilde{\mu}=\frac{1}{m}\sum_{j=1}^m q_U(x_j).
\end{equation}
By \Cref{eq:2nd_moment}, \Cref{eq:3rd_moment}, and Chebyshev's inequality,
\begin{align}
&\Pr_{x_1,\ldots,x_m\sim q_U}\left(\tilde{\mu}\leq \frac{1}{4}\Exp_{V\sim\mathcal{B}}\left[\Exp_{x\sim q_V}[q_V(x)]\right]\right)\\
\leq& \Pr_{x_1,\ldots,x_m\sim q_U}\left(\tilde{\mu}\leq \frac{1}{2}\Exp_{x\sim q_U}[q_U(x)]\right)\\
\leq& \Pr_{x_1,\ldots,x_m\sim q_U}\left(\left|\tilde{\mu}-\Exp_{x\sim q_U}[q_U(x)]\right|\geq\frac{1}{2}\Exp_{x\sim q_U}[q_U(x)]\right)\\
\leq& \frac{4\Var_{x\sim q_U}\left[q_U(x)\right]}{m\left(\Exp_{x\sim q_U}[q_U(x)]\right)^2}\\
\leq& \frac{2400\Exp_{U}\left[\Exp_{x\sim q_U}\left[q_U(x)^2\right]\right]}{m\left(\Exp_{U}\left[\Exp_{x\sim q_U}[q_U(x)]\right]\right)^2}\\
\leq&\frac{3600(1+3e^{-24JT})^n}{m(1+e^{-24JT})^{2n}}.
\end{align}
Therefore, by taking
\begin{equation}
m=\frac{180000(1+3e^{-24JT})^n}{(1+e^{-24JT})^{2n}}
\end{equation}
string samples, we can guarantee that with probability at least $\frac{1}{50}$ over the choice of $U$ and samples $x_1,\ldots,x_m\sim q_U$, we get
\begin{equation}
\tilde{\mu}=\frac{1}{m}\sum_{j=1}^m q_U(x_j)\geq \frac{1}{4}\Exp_V\left[\Exp_{x\sim q_V}\left[q_V(x)\right]\right]=\frac{1}{2\cdot 2^n}(1+e^{-24JT})^n.
\end{equation}
While the sample complexity for the number of strings scales exponentially in the number of qubits $n$ for constant time $T$, it can also be viewed as scaling linearly with the expected \textsf{XEB} score itself. Hence, in a goldilock regime where the expected \textsf{XEB} score is modest, so is the sample complexity.

\section{Numerics on all-to-all random quantum circuits}

\begin{figure}[ht]
\includegraphics[width=\linewidth]{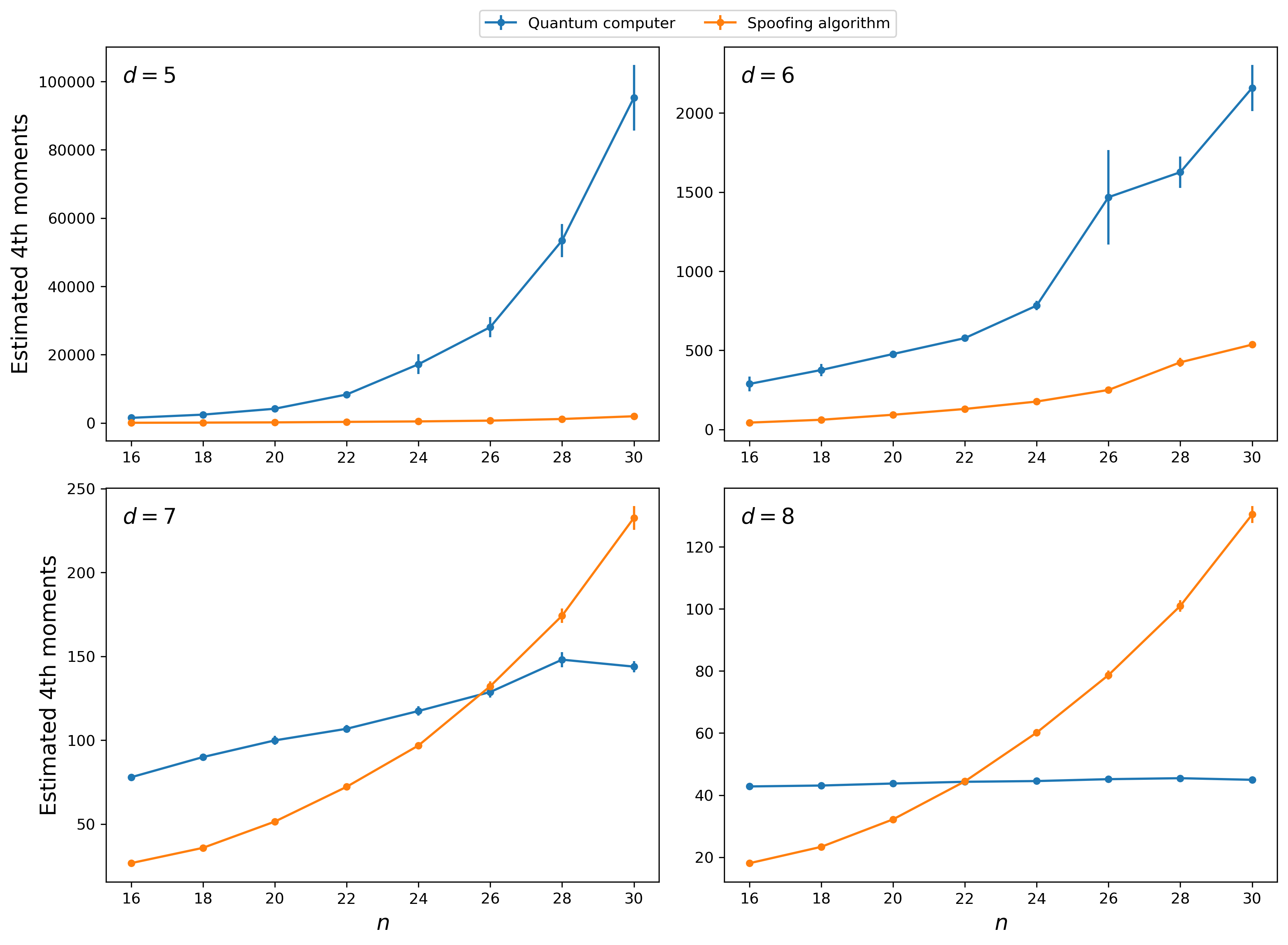}
\caption{Results for numerically estimating \Cref{eq:quantum_numerics} and \Cref{eq:harvard_numerics} for $n$ ranging from $16$ to $30$ and $d$ ranging from $5$ to $8$. Each data point is calculated by an average of $2000$ circuit instances. The error bars show one standard error. The simulations are performed using Cirq \cite{cirq} and qsim \cite{qsim}.}
\label{fig:numerics_all_to_all}
\end{figure}

We have established that for sufficiently large constant-time Brownian circuits, the quantum computer achieves typical \textsf{XEB} scores within a constant factor of the expectation whereas the spoofing algorithm of \cite{gao2024limitations}'s typical performance is troubled by an exponentially larger variance. More specifically, we have calculated the dominant 4th moment quantities
% \begin{equation}
% \Exp_{U\sim\mathcal{B}}\left[\left(\Exp_{x\sim q_U}[q_U(x)]\right)^2\right]=\sum_{x,y\in\{0,1\}^n}\Exp_{U\sim\mathcal{B}}\left[q_U(x)^2q_U(y)^2\right]= \frac{4!}{2^{2n}} \left( 1 + e^{-24 J T} \right)^{2n}
% \end{equation}
\begin{equation}
\sum_{x\in\{0,1\}^n}\Exp_{U\sim\mathcal{B}}\left[q_U(x)^4\right]= \frac{4!}{2^{3n}} \left( 1 + 6e^{-24 J T} +e^{-48JT}\right)^{n}
\end{equation}
for the quantum computer and
% \begin{equation}
% \Exp_{U\sim\mathcal{B}}\left[\left(\Exp_{x\sim A_U}[q_U(x)]\right)^2\right]\geq \sum_{x\in\{0,1\}^n}\Exp_{U\sim\mathcal{B}}\left[A_U(x)^2q_U(x)^2\right]=\frac{2^{K+1}}{2^{2n}}(1+2e^{-24JT})^n
% \end{equation}
\begin{equation}
\sum_{x\in\{0,1\}^n}\Exp_{U\sim\mathcal{B}}\left[A_U(x)^2q_U(x)^2\right]=\frac{2^{K+1}}{2^{3n}}(1 + 6e^{-24 J T} +e^{-48JT})^n
\end{equation}
for the algorithm of \cite{gao2024limitations}. Recall that $q_U$ is the output distribution of the circuit $U$, $A_U$ is the distribution sampled by the spoofing algorithm on input $U$, and $K$ is the number of disjoint qubit subsets created during the algorithm. The distinction in typical behaviors stems from the difference between the constant $4!=24$ term versus the $2^{K+1}$ term which scales exponentially in $n$ since $K\geq\Omega(n/\log n)$. 

In this section, we numerically investigate whether the predictions made through Brownian circuits are reflected in constant-depth all-to-all random quantum circuits. More specifically, we try to estimate the 4th moment quantity
\begin{equation}
2^{3n}\sum_{x\in\{0,1\}^n}\Exp_{U\sim\mathcal{D}_{n,d}}\left[q_U(x)^4\right]
\label{eq:quantum_numerics}
\end{equation}
for the quantum computer and the corresponding quantity
\begin{equation}
2^{3n}\sum_{x\in\{0,1\}^n}\Exp_{U\sim\mathcal{D}_{n,d}}\left[A_U(x)^2q_U(x)^2\right]
\label{eq:harvard_numerics}
\end{equation}
for the algorithm of \cite{gao2024limitations} for $n$ ranging from $16$ to $30$ and $d$ ranging from $5$ to $8$. The spoofing algorithm is implemented using the greedy partition strategy described in \Cref{app:harvard_discrete}. See \Cref{fig:numerics_all_to_all} for plots of our results. We observe that starting from $d=7$, consistent with the large $T$ assumption required by our Brownian analysis, the 4th moment for the spoofing algorithm overtakes that of the quantum computer. Hence, the numerics suggest that this phenomenon predicted by the Brownian circuits does translate to the realm of constant-depth all-to-all random quantum circuits.

\begin{figure}[ht]
\includegraphics[width=\linewidth]{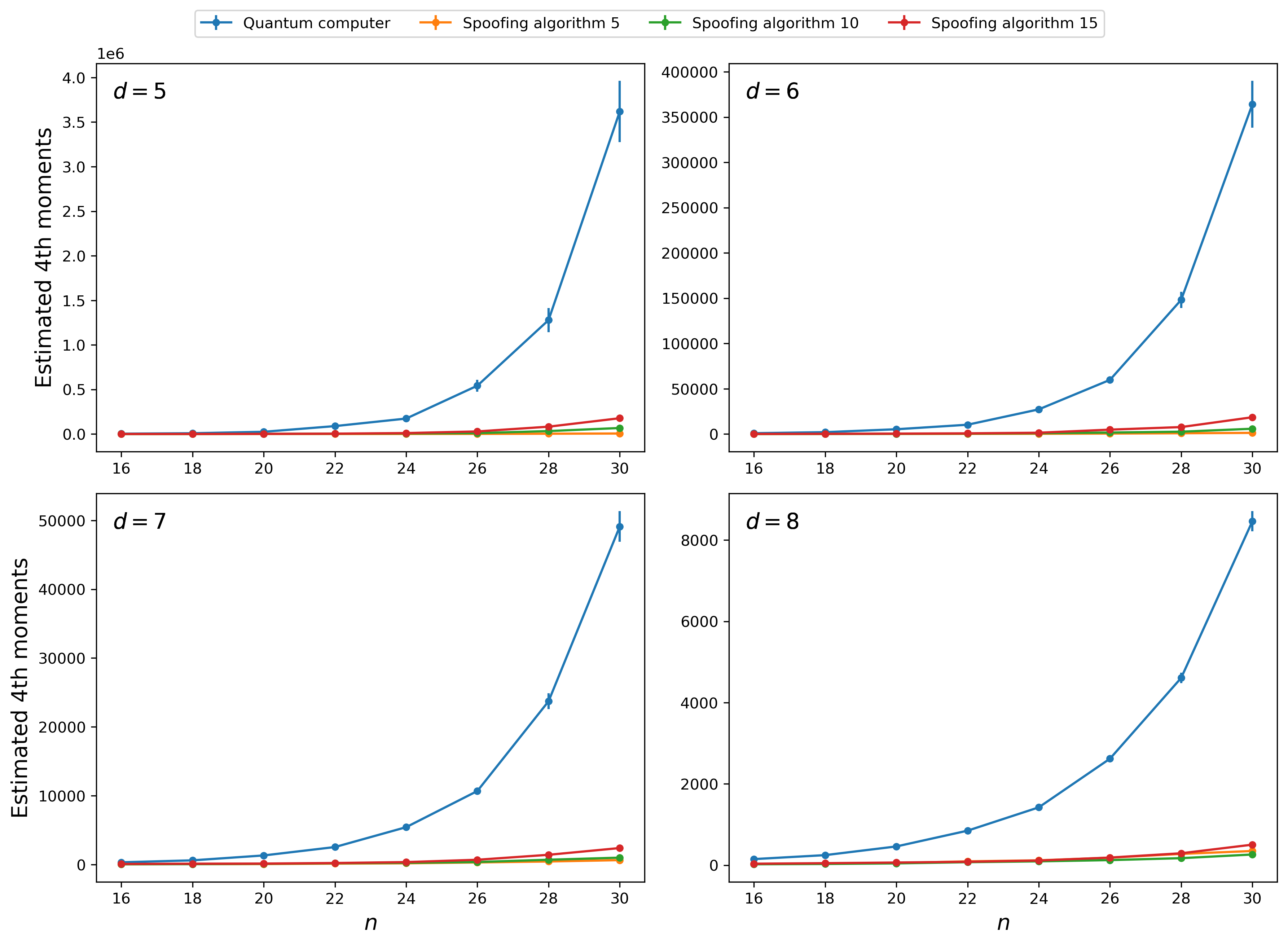}
\caption{Results for numerically estimating \Cref{eq:quantum_numerics} and \Cref{eq:harvard_numerics} for the 1D architecture for $n$ ranging from $16$ to $30$ and $d$ ranging from $5$ to $8$. In the caption, ``Spoofing algorithm $r$'' with $r\in\{5,10,15\}$ indicates that the algorithm partitions the qubits into $K=\lceil\frac{n}{r}\rceil$ subsets where each subset, except maybe the last one, has size $r$. Each data point is calculated by an average of $2000$ circuit instances. The error bars show one standard error. The simulations are performed using Cirq \cite{cirq} and qsim \cite{qsim}.}
\label{fig:numerics_1D}
\end{figure}

To further investigate the ``universality'' of the Brownian dynamics, we have repeated the same numerical simulations on 1D random quantum circuits with periodic boundary conditions (so that every layer again consists of exactly $\frac{n}{2}$ gates). For the 1D architecture, we adopt the natural partition strategy of taking contiguous blocks of qubits as the subsets, for block sizes of $r=5,10,15$. Note that the last subset contains fewer qubits if $r$ does not divide $n$, and the number of subsets $K$ equals $\lceil\frac{n}{r}\rceil$. See \Cref{fig:numerics_1D} for plots of our results for the 1D architecture. Here, we see that the 4th moment quantities are always larger for the quantum computer, and for the spoofing algorithm, the values appear the largest when $r=15$, which correspond to the smallest $K=\lceil\frac{n}{r}\rceil$. Therefore, the numerics suggest that the Brownian predictions are incompatible with 1D random quantum circuits. 

However, this leaves open the possibility that Brownian circuits become accurate models of $D$-dimensional random quantum circuits for sufficiently large constant $D>1$. We formalize this observation in the following conjecture.

\begin{conjecture}
For sufficiently large constant $D$ and $d$, the spoofing algorithm of \cite{gao2024limitations} has exponentially larger variance in $n$ compared to a perfect quantum computer when run on input $D$-dimensional $n$-qubit depth-$d$ brickwork random quantum circuits.
\label{conj:D-dimensional}
\end{conjecture}

\FloatBarrier

\section{Discussion}
\label{sec:conclusion}

In this paper we studied a family of random Brownian all-to-all circuits and used mean-field (large-$n$) techniques to calculate arbitrary moments of the output distribution at constant depth. Access to these higher moments allowed us to characterize not just the expected linear cross-entropy (\textsf{XEB}) benchmark scores but also the typicality of such scores. In particular, we found that the spoofing algorithm suggested in \cite{gao2024limitations} suffers from an exponentially large variance in the number of disjoint subsystems, meaning that the average score is \emph{not} typical. Our numerical simulations of all-to-all random quantum circuits support this conclusion. Our results leave many questions unanswered. In particular, while we used Brownian circuits as a proxy for all-to-all discrete Haar-random unitary circuits, a precise mathematical connection between these two circuit families remains unclear. It is also unclear whether the \XEB~benchmark for Brownian circuits is robust against other types of classical spoofing algorithms. More generally, deeper connections between Brownian circuit sampling and complexity theory remains to be worked out. We leave investigation of these questions for future work. Another open question is to prove or disprove \Cref{conj:D-dimensional}.

\section{Acknowledgements}
B.F. and S.G. ~acknowledge support from AFOSR
(FA9550-21-1-0008).  This material is based upon work partially
supported by the National Science Foundation under Grant CCF-2044923
(CAREER), by the U.S. Department of Energy, Office of Science,
National Quantum Information Science Research Centers (Q-NEXT) and by
the DOE QuantISED grant DE-SC0020360.  This work was done in part while a subset of the authors were visiting the Simons Institute for the Theory of Computing, supported by NSF QLCI Grant No. 2016245.
\bibliographystyle{unsrt}
\bibliography{References}

\pagebreak
\newpage

\appendix

\section{The spoofing algorithm of \cite{gao2024limitations} for all-to-all random quantum circuits}
\label{app:harvard_discrete}

We first describe a simple greedy partition strategy.

\begin{figure}[ht]
\begin{minipage}{0.75\linewidth}
\begin{algorithm}[H]
	\caption{Greedy partition for the all-to-all architecture}
 \label{alg:greedy_partition}
	\hspace*{\algorithmicindent} \hspace{-26pt} \textbf{Input:}  an $n$-qubit depth-$d$ RQC $U\sim\mathcal{D}_{n,d}$\\
 \hspace*{\algorithmicindent} \hspace{-28pt} \textbf{Output:} a partition $\hat{\mathcal{P}}$ of $\{1,\ldots,n\}$
\begin{algorithmic}[1]
\State{$G\gets\{1,\ldots,n\}$}
\State{$\mathcal{P}\gets\emptyset, \tilde{\mathcal{P}}\gets\emptyset$}
\While{$G\neq\emptyset$}
\State{Choose the smallest $i\in G$}
\State{$P_i\gets\{i\}\cup\{j:\text{there is a gate in $U$ acting on qubits $i$ and $j$}\}$}
\State{$\mathcal{P}\gets\mathcal{P}\cup\{P_i\}$, $\tilde{\mathcal{P}}\gets\tilde{\mathcal{P}}\cup P_i$}
\State{$D_i\gets\{j:\text{there is a gate in $U$ acting on qubit $j$ and some qubit $k\in P_i$}\}$}
\State{$G\gets G\setminus D_i$}
\EndWhile
\State{$\hat{\mathcal{P}}\gets\mathcal{P}\cup\{\{j\}:j\in\{1,\ldots,n\}\setminus\tilde{P}\}$}
\State{\textbf{return} $\hat{\mathcal{P}}$}
\end{algorithmic}
\end{algorithm}
\end{minipage}
\end{figure}

\begin{theorem}
The following properties hold for \Cref{alg:greedy_partition}:
\begin{enumerate}
\item At the end of each iteration of the while loop, $|P_i|\leq d+1$ and $|D_i|\leq d^2+1$. The while loop terminates after at most $n$ iterations.

\item After exiting the while loop, $|\mathcal{P}|\geq\frac{n}{d^2+1}$.

\item After exiting the while loop, for every $P_i\in\mathcal{P}$, for every gate $V$ in $U$ acting on the qubit $i$ and some other qubit $j$, $j\in P_i$. 

\item After exiting the while loop, for every $P_i,P_{i'}\in\mathcal{P}$ with $i\neq i'$, $P_i\cap P_{i'}=\emptyset$, so $\hat{\mathcal{P}}$ is a partition of $\{1,\ldots,n\}$. 
\end{enumerate}
\end{theorem}

\begin{proof}
\text{ }
\begin{enumerate}
\item Consider an iteration of the while loop where $i\in G$ is chosen. Since $U$ has depth $d$, by the definition of $\mathcal{D}_{n,d}$, there are exactly $d$ $2$-qubit gates acting on the qubit $i$. In total, these gates can act on at most $d+1$ distinct qubits, so $|P_i|\leq d+1$. Since $i\in P_i$, we get $P_i\subseteq D_i$. Suppose there exists $j\in D_i\setminus P_i$. Then there exists $k\in P_i\setminus\{i\}$ such that there is a gate acting on qubits $j$ and $k$. We have $|P_i\setminus\{i\}|\leq d$ and since $k\in P_i\setminus\{i\}$, there is a gate acting on qubits $i$ and $k$. Thus, $|D_i|=|P_i|+|D_i\setminus P_i|\leq d+1+d(d-1)=d^2+1$. Since $i\in P_i\subseteq D_i$ and $i\in G$, the size of $G$ decreases by at least one in each iteration. Thus, the while loop terminates after at most $n$ iterations.

\item Since $|D_i|\leq d^2+1$ in every iteration, the while loop is executed at least $\frac{n}{d^2+1}$ times, so $|\mathcal{P}|\geq\frac{n}{d^2+1}$.

\item This property holds by the definition of each $P_i\in\mathcal{P}$.

\item Let $P_i,P_{i'}\in\mathcal{P}$ with $i\neq i'$. WLOG, suppose $i<i'$, so $P_i$ is added to $\mathcal{P}$ before $P_{i'}$ is. Assume by contradiction that there exists $k\in P_i\cap P_{i'}$. Since $k\in P_i$ and there is a gate acting on qubits $k$ and $i'$, $i'\in D_i$. Thus, $i'$ is deleted from $G$ in the iteration $P_i$ is added to $\mathcal{P}$. Clearly, we have a contradiction.
\end{enumerate}
\end{proof}

Next, we prove that the \textsf{XEB} score achieved by \Cref{def:HA} increases exponentially in $n$ in the constant-depth regime.

\begin{theorem}
Let $A$ denote the spoofing algorithm of \cite{gao2024limitations} using the greedy partition strategy described in \Cref{alg:greedy_partition}. For every even $n\geq 2$ and $d\geq 1$, it holds that
\begin{equation}
\Exp_{U\sim\mathcal{D}_{n,d}}[\textsf{XEB}(U,A)]\geq \left(1+\left(\frac{1}{15}\right)^d\right)^{\frac{n}{d^2+1}}.
\end{equation}
\end{theorem}
The proof relies on two basic facts about the first and second moments of conjugating a Pauli operator by a Haar-random $2$-qubit gate.

\begin{fact}\label{L:first_moment_formula}
For every $P\in\{I,X,Y,Z\}^{\otimes 2}$,
$$\Exp_{U\sim\mathbb{U}(4)}[UPU^\dagger]=\begin{cases}
I\otimes I & \text{if $P=I\otimes I$,}\\
0 & \text{otherwise.}
\end{cases}$$
\end{fact}

\begin{fact}\label{L:second_moment_formula}
For every $P\in\{I,X,Y,Z\}^{\otimes 2}$,
$$\Exp_{U\sim\mathbb{U}(4)}[(U\otimes U)(P\otimes P)(U^\dagger\otimes U^\dagger)]=\begin{cases}
I\otimes I\otimes I\otimes I & \text{if $P=I\otimes I$,}\\
\frac{1}{15}\sum_{Q\in\{I,X,Y,Z\}^{\otimes 2}\setminus\{I\otimes I\}} Q\otimes Q & \text{otherwise.}
\end{cases}$$
\end{fact}

\begin{proof}
Suppose the input circuit $U$ consists of $s$ $2$-qubit gates, so we can write $U$ as $U=U_s\ldots U_1$. Let $O\subseteq[s]$ denote the set of indices of gates replaced by \Cref{def:HA}. Let $\bar{U}=\bar{U}_s\cdots \bar{U}_1$ denote the modified circuit constructed by the algorithm so that for every $j\in[s]$, $\bar{U}_j=U_j$ if $j\notin O$ and $\bar{U}_j=I\otimes I$ if $j\in O$. Then we can write
\begin{align*}
\Exp_{U\sim\mathcal{D}_{n,d}}[\textsf{XEB}(U,A)]&=2^n\sum_{z\in\{0,1\}^n}\Exp_{U\sim\mathcal{D}_{n,d}}\left[|\bra{z}U\ket{0^n}|^2\cdot|\bra{z}\bar{U}\ket{0^n}|^2\right]\\
&=2^n\sum_{z\in\{0,1\}^n}\Exp_{U\sim\mathcal{D}_{n,d}}\left[\bra{0^{2n}}(U^\dagger\otimes \bar{U}^\dagger)(\ket{z}\bra{z}\otimes\ket{z}\bra{z})(U\otimes\bar{U})\ket{0^{2n}}\right]\\
&=\sum_{w\in\{0,1\}^n}\Exp_{U\sim\mathcal{D}_{n,d}}\left[\bra{0^{2n}}(U^\dagger\otimes \bar{U}^\dagger)(Z(w)\otimes Z(w))(U\otimes\bar{U})\ket{0^{2n}}\right]
\end{align*}
where we use the notation $Z(w)=\bigotimes_{j=1}^n Z^{w_j}$ for every $w\in\{0,1\}^n$.

Let $w\in\{0,1\}^n$. The above expectation for $w$ can be written as a sum over trajectories that can be enumerated inductively. Each trajectory is of the form $\vec{\gamma}_s=(\gamma_0,\ldots,\gamma_s)$ with associated weights $c(\vec{\gamma}_s)\geq 0$ such that for every $j\in[s]\cup\{0\}$, either $\gamma_j\in\{I,X,Y,Z\}^{\otimes n}$ or $\gamma_j=0$. The trajectories will satisfy
$$\Exp_{U\sim\mathcal{D}_{n,d}}\left[(U^\dagger\otimes \bar{U}^\dagger)(Z(w)\otimes Z(w))(U\otimes\bar{U})\right]=\sum_{\vec{\gamma}_s}c(\vec{\gamma}_s)\gamma_s\otimes\gamma_s$$
where the sum is over all the trajectories. For the base case, we start with $\gamma_0=Z(w)$, $\vec{\gamma}_0=(\gamma_0)$, and $c(\vec{\gamma}_0)=1$. Let $j\in[s-1]\cup\{0\}$, and suppose we have constructed a trajectory $\vec{\gamma}_j=(\gamma_0,\ldots,\gamma_{j})$ by induction with weight $c(\vec{\gamma}_j)\geq 0$. If $\gamma_j=0$, then we have $\vec{\gamma}_{j+1}=(\gamma_0,\ldots,\gamma_{j},\gamma_{j+1})$ with $\gamma_{j+1}=0$ and $c(\vec{\gamma}_{j+1})=0$. Now suppose $\gamma_j\in\{I,X,Y,Z\}^{\otimes n}$. Suppose the $(j+1)$-st gate acts on qubits $a$ and $b$. We first consider the case where $j+1\notin O$, so the non-trivial action of the $(j+1)$-st gate is
$$\Exp_{U\sim\mathbb{U}(4)}[(U^\dagger\otimes U^\dagger)((\gamma_j)_a\otimes(\gamma_j)_b\otimes(\gamma_j)_a\otimes(\gamma_j)_b)(U\otimes U)]$$
By \Cref{L:second_moment_formula}, if $(\gamma_j)_a\otimes (\gamma_j)_b=I\otimes I$, then we have $\vec{\gamma}_{j+1}=(\gamma_0,\ldots,\gamma_j,\gamma_{j+1})$ where $\gamma_{j+1}=\gamma_j$ and $c(\vec{\gamma}_{j+1})=c(\vec{\gamma}_j)\geq 0$. If $(\gamma_j)_a\otimes (\gamma_j)_b\neq I\otimes I$, then by \Cref{L:second_moment_formula}, the expectation over the $(j+1)$-st gate splits $\vec{\gamma}_j$ into 15 trajectories as follows. For every $Q_1\otimes Q_2\in\{I,X,Y,Z\}^{\otimes 2}\setminus\{I\otimes I\}$, we have a $\gamma_{j+1}\in\{I,X,Y,Z\}^{\otimes n}$ where $(\gamma_{j+1})_a=Q_1$, $(\gamma_{j+1})_b=Q_2$, and for every $k\notin\{a,b\}$, $(\gamma_{j+1})_k=(\gamma_j)_k$, and the trajectory $\vec{\gamma}_{j+1}=(\gamma_0,\ldots,\gamma_j,\gamma_{j+1})$ formed has weight $c(\vec{\gamma}_{j+1})=\frac{1}{15}c(\vec{\gamma}_{j})\geq 0$. If $j+1\neq s$, then the next gate will act on the 15 spawned trajectories separately by the linearity of expectation. Now suppose $j+1\in O$. Then for the $(j+1)$-st gate, we need to evaluate
$$\Exp_{U\sim\mathbb{U}(4)}\left[(U^\dagger\otimes I)((\gamma_j)_a\otimes(\gamma_j)_b\otimes(\gamma_j)_a\otimes(\gamma_j)_b)(U\otimes I)\right]$$
By \Cref{L:first_moment_formula}, if $(\gamma_j)_a\otimes (\gamma_j)_b=I\otimes I$, then we have $\vec{\gamma}_{j+1}=(\gamma_0,\ldots,\gamma_j,\gamma_{j+1})$ where $\gamma_{j+1}=\gamma_j$ and $c(\vec{\gamma}_{j+1})=c(\vec{\gamma}_j)\geq 0$. Otherwise, we have
$$\Exp_{U\sim\mathbb{U}(4)}\left[(U^\dagger\otimes I)((\gamma_j)_a\otimes(\gamma_j)_b\otimes(\gamma_j)_a\otimes(\gamma_j)_b)(U\otimes I)\right]=0,$$
so we get $\vec{\gamma}_{j+1}=(\gamma_0,\ldots,\gamma_j,\gamma_{j+1})$ with $\gamma_{j+1}=0$ and $c(\vec{\gamma}_{j+1})=0$.

Overall, we have that
$$\sum_{w\in\{0,1\}^n}\Exp_{U\sim\mathcal{D}_{n,d}}\left[\bra{0^{2n}}(U^\dagger\otimes \bar{U}^\dagger)(Z(w)\otimes Z(w))(U\otimes\bar{U})\ket{0^{2n}}\right]=\sum_{\vec{\gamma}_s}c(\vec{\gamma}_s)\bra{0^{2n}}(\gamma_s\otimes\gamma_s)\ket{0^{2n}},$$
so each trajectory $\vec{\gamma}_s$ contributes $c(\vec{\gamma}_s)\bra{0^{2n}}(\gamma_s\otimes\gamma_s)\ket{0^{2n}}\geq 0$ to the sum since $\bra{0^{2n}}(\gamma_s\otimes\gamma_s)\ket{0^{2n}}\in\{0,1\}$ for every $\gamma_s\in\{I,X,Y,Z\}^{\otimes n}\cup\{0\}$. Therefore, we can count the contributions from a subset of all the trajectories to arrive at a lower bound. 

Notice that for every partition $P_i\in\mathcal{P}$ and every gate $U_j$ in $U$ acting on qubit $i$, $\bar{U}_j=U_j$. Define $R=\{i\in[n]:P_i\in\mathcal{P}\}$ and $W=\{w\in\{0,1\}^n:w_i=0\text{ }\forall i\notin R\}$. Then for every $w\in W$, the trajectory $\vec{\gamma}_s^*=(\gamma_0^*,\ldots,\gamma_s^*)$ where $\gamma_0^*=\cdots=\gamma_s^*=Z(w)$ exists and $c(\vec{\gamma}_s^*)\geq (\frac{1}{15})^{|w|d}$. Therefore,
\begin{equation}
\Exp_{U\sim\mathcal{D}_{n,d}}[\textsf{XEB}(U,A)]\geq\sum_{k=0}^{|\mathcal{P}|}\binom{|\mathcal{P}|}{k}\left(\frac{1}{15}\right)^{kd}=\left(1+\left(\frac{1}{15}\right)^d\right)^{|\mathcal{P}|}\geq\left(1+\left(\frac{1}{15}\right)^d\right)^{\frac{n}{d^2+1}}.
\end{equation}
\end{proof}

\section{Brownian Circuit Calculations}

Here we derive expressions for moments of the transition probability $q_U(\mathbf{x}) = \magn{\bra{\mathbf{x}} U \ket{\mathbf{0}}}^2$ for Brownian circuits. In what follows, $\mathbb{E} = \mathbb{E}_{U \sim \mathcal{B}}$ is always understood to mean the expectation over Brownian circuits. We begin with relatively simple expressions and add complexity step by step as the final calculations can become somewhat complicated. Because there are many equations, we list the main results here:
\begin{itemize}
    \item The exact expectation value of the return probability $\mathbb{E}[q_U(\mathbf{0})]$ appears in Eq. \eqref{eq:p0expecval}.
    \item The exact expectation value of the transition probability $\mathbb{E}[q_U(\mathbf{x})]$ appears in Eq. \eqref{eq:k1final}. We perform the sum over bitstrings in Eq. \eqref{eq:k1sumbitstrings}.
    \item A saddle-point calculation of the transition probability $\mathbb{E}[q_U(\mathbf{x})]$ appears in Eq. \eqref{eq:k1saddlepoint}.
    \item A saddle-point calculation of the $k$th-order transition probability $\mathbb{E}[q_U^k(\mathbf{x})]$ appears in Eq. \eqref{eq:arbitrarykfinal}. We perform the sum over bitstrings in Eq. \eqref{eq:arbitraryksumbitstrings}.
    \item A large-$n$ calculation of the expected overlap $\mathbb{E}[q_U(\mathbf{x}) A_U(\mathbf{x})]$ appears in Eq. \eqref{eq:pqoverlap}, where $A_U(\mathbf{x})$ is the disjoint circuit used in the algorithm of \cite{gao2024limitations}. We perform the sum over bitstrings in Eq. \eqref{eq:pqoverlapsumbits} by making further large-$n$ approximations.
    \item A large-$n$ calculation of the $c$th moment of the expected overlap $\mathbb{E}[q_U^c(\mathbf{x}) A_U^c(\mathbf{x})]$ appears in Eq \eqref{eq:pqcthmoment}, where $k = 2c$. We perform the sum over bitstrings in Eq. \eqref{eq:pqcthmomentsumbitstrings}.
    \item A comparison of the \textsf{XEB} variance for the spoofing algorithm of \cite{gao2024limitations} versus the true distribution appears in Eq. \eqref{eq:varspoof} and \eqref{eq:vartrue}.
\end{itemize}

\pagebreak
\newpage

\subsection{First Moment}

We start by evaluating expectation values of the first moment of the transition probability, corresponding to $k = 1$.

\subsubsection{Return Probability}

The expectation value of the return probability for a single all-to-all Brownian circuit on $n$ qubits is
\begin{equation}
	\mathbb{E}[q_U(\mathbf{0})] = \mathbb{E}\left[\magn{\bra{\mathbf{0}} U \ket{\mathbf{0}}}^2 \right]
\end{equation}
where we use boldface to denote bitstring vectors $\mathbf{0} \equiv 0^n$ and where the expectation value is taken over circuit realizations. Here the Brownian unitary evolution $U$ is defined by:
\begin{equation}
	U = \prod_{t=0}^T U_t = \prod_t \exp \left(-i \sum_{\substack{i<j \\ \alpha\beta}} J_{ij}^{\alpha \beta}(t) \sigma_i^{\alpha} \sigma_j^{\beta} \Delta t \right)
\end{equation}
where $\sigma_i^{\alpha}$ are the Pauli matrices on site $i$, and we have discretized the time evolution into steps of size $\Delta t$ with the understanding that we will consider the limit $\Delta t \rightarrow 0$ at the end of the calculation. The Brownian coupling coefficients $J_{ij}^{\alpha \beta}(t)$ are Gaussian white-noise variables with zero mean and variance
\begin{equation}
	\mathbb{E} \left[ J_{ij}^{\alpha \beta}(t) J_{i'j'}^{\alpha' \beta'}(t') \right] = \delta_{ii'} \delta_{jj'} \delta^{\alpha \alpha'} \delta^{\beta \beta'} \delta_{tt'} \frac{J}{n \Delta t}.
\end{equation}

It is convenient to write the return probability as a doubled system with two `replicas,' which we will subsequently refer to as $L$ and $R$:
\begin{equation}
	q_U(\mathbf{0}) = \bra{\mathbf{0}} U \ket{\mathbf{0}} \left(\bra{\mathbf{0}} U \ket{\mathbf{0}}\right)^* = \bra{\mathbf{00}} U \otimes U^* \ket{\mathbf{00}}.
\end{equation}
Further, the physics of the problem becomes more transparent when we use a proper time-reversal operation instead of the complex conjugate. We define the time-reversal operator $\mathcal{T}$ as:
\begin{equation}
	\mathcal{T}(U) \equiv U^{\mathcal{T}} = (i Y)^{\dagger} U^* (i Y)
\end{equation}
where
\begin{equation}
	iY = \prod_i \left(i \sigma_i^y\right) = \bigotimes_i \begin{bmatrix}
0 & 1 \\
-1 & 0 \\
\end{bmatrix}_i
\end{equation}
and $(i Y)^{\dagger} (i Y) = (i Y) (i Y)^{\dagger} = 1$. Note that the Pauli operators uniformly reverse sign under the time-reversal operation: $\mathcal{T}(\sigma_i^{\alpha}) = - \sigma_i^{\alpha}$, whereas they do not under complex conjugation. Explicitly, bra and ket vectors transform under this operation as follows:
\begin{align}
	\label{eq:timerevbrakets}
	\left(i \sigma^y\right)^{\dagger} \ket{0} &= \ket{1} \nonumber \\
	\left(i \sigma^y \right)^{\dagger} \ket{1} &= - \ket{0} \nonumber \\
	\bra{0} i \sigma^y &= \bra{1} \nonumber \\
	\bra{1} i \sigma^y &= -\bra{0}.
\end{align}
We therefore find the following expression for the return probability:
\begin{align}
	q_U(\mathbf{0}) &= \bra{\mathbf{00}} U \otimes (iY) (iY)^{\dagger} U^* (iY) (iY)^{\dagger} \ket{\mathbf{00}} \nonumber \\
		&= \bra{\mathbf{01}} U \otimes U^{\mathcal{T}} \ket{\mathbf{01}} \nonumber \\
		&= \bra{\mathbf{01}} \prod_t U_t \otimes U_t^{\mathcal{T}} \ket{\mathbf{01}}
\end{align}

Next we consider the average over circuit realizations. Because each timestep of the circuit features independent Brownian variables, we may evaluate the expectation value at each timestep independently:
\begin{equation}
	\mathbb{E}[q_U(\mathbf{0})] = \mathbb{E}\left[ \bra{\mathbf{01}} \prod_t U_t \otimes U_t^{\mathcal{T}} \ket{\mathbf{01}} \right] = \bra{\mathbf{01}} \prod_t \mathbb{E}\left[ U_t \otimes U_t^{\mathcal{T}} \right] \ket{\mathbf{01}}.
\end{equation}
To evaluate the disorder average, we first expand the exponential in $U_t \otimes U_t^{\mathcal{T}}$ using the Baker-Campbell-Hausdorff formula and keep only the first-order terms with the expectation that we will take the limit $\Delta t \rightarrow 0$ at the end of the calculation:
\begin{align}
    U_t \otimes U_t^{\mathcal{T}} &= \exp \left[-i \sum_{\substack{i<j \\ \alpha\beta}} J_{ij}^{\alpha \beta}(t) \left( \sigma_{iL}^{\alpha} \sigma_{jL}^{\beta} - \sigma_{iR}^{\alpha} \sigma_{jR}^{\beta} \right) \Delta t \right] \nonumber \\
    & \approx \prod_{\substack{i<j \\ \alpha\beta}} \exp \left[-i J_{ij}^{\alpha \beta}(t) \left( \sigma_{iL}^{\alpha} \sigma_{jL}^{\beta} - \sigma_{iR}^{\alpha} \sigma_{jR}^{\beta} \right) \Delta t \right]
\end{align}
where we use the notation $L,R$ to refer to spins in the first and second replica, respectively. Then, because the Brownian couplings are Gaussian white-noise variables, we can write the disorder-average expectation value as a product of Gaussian integrals:
\begin{equation}
	\mathbb{E}\left[ U_t \otimes U_t^{\mathcal{T}} \right] = \prod_{\substack{i<j \\ \alpha \beta}} \left( \int d J_{ij}^{\alpha \beta}(t) \exp\left[ - \frac{\left( J_{ij}^{\alpha \beta}(t) \right)^2 }{2J/\Delta t n} \right] \exp \left[ - i J_{ij}^{\alpha \beta}(t) \Delta t \left( \sigma_{iL}^{\alpha} \sigma_{jL}^{\beta} - \sigma_{iR}^{\alpha} \sigma_{jR}^{\beta}  \right) \right] \right).
\end{equation}
We may immediately evaluate these Gaussian integrals to obtain (again, to lowest order in $\Delta t$):
\begin{equation}
	\mathbb{E}\left[ U_t \otimes U_t^{\mathcal{T}} \right] = \exp \left[ - \frac{J}{2N} \sum_{\substack{i<j \\ \alpha \beta}} \left( \sigma_{iL}^{\alpha} \sigma_{jL}^{\beta} - \sigma_{iR}^{\alpha} \sigma_{jR}^{\beta} \right)^2 \Delta t \right].
\end{equation}
By stacking these individual timesteps up again sequentially, we therefore find that the expectation value of the return probability can be written in terms of imaginary-time evolution
\begin{equation}
	\mathbb{E}[q_U(\mathbf{0})] = \bra{\mathbf{01}} e^{-H_{\mathrm{eff}} T} \ket{\mathbf{01}}
\end{equation}
with an effective Hamiltonian
\begin{equation}
	H_{\mathrm{eff}} = \frac{J}{2N} \sum_{\substack{i<j \\ \alpha \beta}} \left( \sigma_{iL}^{\alpha} \sigma_{jL}^{\beta} - \sigma_{iR}^{\alpha} \sigma_{jR}^{\beta} \right)^2.
    \label{eq:k1effectiveham}
\end{equation}
After some manipulation, this effective Hamiltonian can be rewritten entirely in terms of Heisenberg interactions:
\begin{equation}
    H_{\mathrm{eff}} = - \frac{J}{2 n} \left( \sum_i \vec{\sigma}_{iL} \cdot \vec{\sigma}_{iR} \right)^2 + \frac{9 J}{2} (n-1) + \frac{J}{2 n} \sum_i \left( \vec{\sigma}_{iL} \cdot \vec{\sigma}_{iR} \right)^2
	\label{eq:effhamk2}
\end{equation}

\subsubsection{Effective Hamiltonian Spectrum}
\label{sec:exactspectrum}

Due to its simple form and mutually-commuting Heisenberg interactions, the effective Hamiltonian \eqref{eq:effhamk2} is exactly solvable. First, note that the Heisenberg interaction is easily diagonalized:
\begin{equation}
	\vec{\sigma}_{iL} \cdot \vec{\sigma}_{iR} = -3 \ket{s}\bra{s}_i + \ket{t}\bra{t}_i + \ket{11} \bra{11}_i + \ket{00} \bra{00}_i
\end{equation}
where
\begin{align}
	\ket{s} &= \frac{1}{\sqrt{2}}(\ket{01}-\ket{10})_{LR} \nonumber \\
	\ket{t} &= \frac{1}{\sqrt{2}}(\ket{01}+\ket{10})_{LR} \nonumber \\
\end{align}
are the singlet and triplet states that, together with the vectors $\ket{00},\ket{11}$, span the two-qubit Hilbert space for $L,R$. Similarly,
\begin{equation}
	\left( \vec{\sigma}_{iL} \cdot \vec{\sigma}_{iR} \right)^2 = 9 \ket{s}\bra{s}_i + \ket{t}\bra{t}_i + \ket{11} \bra{11}_i + \ket{00} \bra{00}_i
\end{equation}
Therefore, the eigenvectors of $H_{\mathrm{eff}}$ are product states of the form
\begin{equation}
	\ket{\Psi} = \bigotimes_i \ket{\psi}_{iL,iR}
\end{equation}
where each $\ket{\psi}_{iL,iR}$ can independently take one of the values $\ket{s},\ket{t},\ket{00},\ket{11}$. The unique ground state is a tensor product of all singlet states
\begin{equation}
	\ket{\Omega} = \bigotimes_i \ket{s}_{iL,iR}
\end{equation}
while higher-energy states have excitations $\ket{t}_i,\ket{00}_i,\ket{11}_i$. The eigenvalues are
\begin{align}
	H_{\mathrm{eff}} \ket{\Psi} &= \left( -\frac{J}{2N} \left[ (-3)(n-z) + z \right]^2 + \frac{9J}{2}(n-1) + \frac{J}{2N} \left[ 9(n-z) + z \right] \right) \ket{\Psi} \nonumber \\
		&= \left( \frac{4 J z (3 n - 2z - 1)}{n} \right) \ket{\Psi}
\end{align}
where the integer $z = 0,1,\ldots,n$ counts the number of excitations ($\ket{t},\ket{00},\ket{11}$) appearing in the state $\ket{\Psi}$.

We can use this exact spectrum to calculate the return probability expectation value. First, we write the state $\ket{\mathbf{01}}$ as a sum over the eigenvectors of $H_{\mathrm{eff}}$. To do so, note that for each site $i$ we can write
\begin{equation}
	\ket{01}_i = \frac{1}{\sqrt{2}} \left( \ket{s}_i + \ket{t}_i \right) = \frac{1}{\sqrt{2}} \sum_{z_i = 0,1} \ket{z_i}
\end{equation}
where the boolean variable $z_i$ indicates whether we have the state $\ket{s}_i = \ket{z_i=0}$ or the state $\ket{t}_i = \ket{z_i=1}$. Applying this notation to the entire state $\ket{\mathbf{01}}$, we find
\begin{equation}
	\ket{\mathbf{01}} \equiv \prod_i \ket{01}_i = \frac{1}{2^{n/2}} \prod_i \left( \ket{s}_i + \ket{t}_i \right) = \frac{1}{2^{n/2}} \sum_{\mathbf{z} = \{0,1\}} \ket{\mathbf{z}}
\end{equation}
where the bitstring $\mathbf{z} = z_1 z_2 \ldots z_N$ indicates which eigenstate we have. For example, the all-zero bitstring $\mathbf{z} = \mathbf{0}$ is the ground state $\ket{\Omega}$. With this notation, each state $\ket{\mathbf{z}}$ is an eigenvector of $H_{\mathrm{eff}}$ with eigenvalue
\begin{equation}
	H_{\mathrm{eff}} \ket{\mathbf{z}} = \left( \frac{4 J \magn{\mathbf{z}} (3 n - 2 \magn{\mathbf{z}} - 1)}{n} \right) \ket{\mathbf{z}}
\end{equation}
where $\magn{\mathbf{z}}$ is the Hamming weight of the bitstring $\mathbf{z}$. Putting this all together, and noting that $\bracket{\mathbf{01}}{\mathbf{z}} = 2^{-n/2}$, we finally find
\begin{equation}
	\mathbb{E}[q_U(\mathbf{0})] = \bra{\mathbf{01}} e^{-H_{\mathrm{eff}} T} \ket{\mathbf{01}} = \frac{1}{2^n} \sum_{\mathbf{z} = \{0,1\}} \exp \left(- \frac{4 J T \magn{\mathbf{z}} (3 n - 2 \magn{\mathbf{z}} - 1)}{n} \right).
	\label{eq:p0expecval}
\end{equation}
Notice that as $T \rightarrow 0$, this expression goes to $\mathbb{E}[q_U(\mathbf{0})] \rightarrow 1$ as expected. For $T \rightarrow \infty$, we have $\mathbb{E}[q_U(\mathbf{0})] \rightarrow 2^{-n}$.
% \gsb{Can probably do this sum for $n \rightarrow \infty$ by considering a continuum limit and a Gaussian integral}

\subsubsection{Transition Probability for Arbitrary Bitstring}

We now generalize the above derivation to calculate the probability $q_U(\mathbf{x})$ of transitioning from $\mathbf{0}$ to an arbitrary bitstring $\mathbf{x}$. Introducing two replicas and using the time-reversal operation as above, we find
\begin{align}
	q_U(\mathbf{x}) &= \magn{\bra{\mathbf{x}} U \ket{\mathbf{0}}}^2 \nonumber \\
		&= \bra{\mathbf{xx}} U \otimes U^* \ket{\mathbf{00}} \nonumber \\
		&= \left[ \bra{\mathbf{xx}} \mathbb{I} \otimes (iY) \right] U \otimes U^{\mathcal{T}} \ket{\mathbf{01}} \nonumber \\
		&= (-1)^{\mathbf{x}} \bra{\mathbf{x\overline{x}}} U \otimes U^{\mathcal{T}} \ket{\mathbf{01}}
	\label{eq:pxderiv}
\end{align}
where in the last line we have used Eq. \eqref{eq:timerevbrakets} and introduced the inverse $\overline{\mathbf{x}}$ of the bitstring $\mathbf{x}$, with $\overline{0} = 1$ and $\overline{1} = 0$. To find the expectation value $\mathbb{E}[q_U(\mathbf{x})]$, the calculation proceeds identically to the derivation above, except that we must now consider the overlap between the states $\ket{\mathbf{x\overline{x}}}$ and $\ket{\mathbf{z}}$. To calculate this, we collect the relevant overlaps in Table \ref{tab:overlaps}
\begin{equation}
\begin{tabular}{c | c | c}
$x_i$ & $z_i$ & overlap \\
\hline
$0$ & $0$ & $\bracket{01}{s} = 1/\sqrt{2}$ \\
$0$ & $1$ & $\bracket{01}{t} = 1/\sqrt{2}$ \\
$1$ & $0$ & $\bracket{10}{s} = -1/\sqrt{2}$ \\
$1$ & $1$ & $\bracket{10}{t} = 1/\sqrt{2}$ \\
\end{tabular}
\label{tab:overlaps}
\end{equation}
We therefore find
\begin{equation}
	\bracket{\mathbf{x\overline{x}}}{\mathbf{z}} = \prod_i \bracket{x_i \overline{x_i}}{z_i} = \frac{1}{2^{n/2}} \prod_i (-1)^{x_i} (-1)^{x_i z_i} = \frac{1}{2^{n/2}} (-1)^{\mathbf{x}} (-1)^{\mathbf{x} \cdot \mathbf{z}}
\end{equation}
where $\mathbf{x} \cdot \mathbf{z} \equiv \left( \sum_i x_i z_i \right) \mod 2$ is the bitwise dot product between bitstrings. We can now collect our results together to evaluate the expectation value $\mathbb{E}[q_U(\mathbf{x})]$. Crucially, the leading factor of $(-1)^{\mathbf{x}}$ cancels with the leading factor in Eq. \eqref{eq:pxderiv}, leaving us with the final result
\begin{equation}
	\mathbb{E}[q_U(\mathbf{x})] = \frac{1}{2^n} \sum_{\mathbf{z} = \{0,1\}} (-1)^{\mathbf{x} \cdot \mathbf{z}} \exp \left(- \frac{4 J T \magn{\mathbf{z}} (3 n - 2 \magn{\mathbf{z}} - 1)}{n} \right).
    \label{eq:k1final}
\end{equation}
Comparing this with Eq. \eqref{eq:p0expecval} we notice that the only difference is the factor of $(-1)^{\mathbf{x} \cdot \mathbf{z}}$, which does not affect the spectrum of $H_{\mathrm{eff}}$. Notice that in the limit $T \rightarrow 0$, this expression goes to $\mathbb{E}[q_U(\mathbf{x})] \rightarrow \delta_{\mathbf{x},\mathbf{0}}$ as expected. Notice also that as $T \rightarrow \infty$, we have $\mathbb{E}[q_U(\mathbf{x})] \rightarrow 2^{-n}$, i.e. at long times all bitstrings $\mathbf{x}$ are equally probable.

Finally, we can use this result to compute the sum over all bitstrings $\mathbf{x}$. We make use of the identity
\begin{equation}
	\sum_{\mathbf{x} = \{0,1\}} (-1)^{\mathbf{x} \cdot \mathbf{z}} = \prod_i \sum_{x_i = 0,1} (-1)^{x_i z_i} = \prod_i (1+(-1)^{z_i}) = 2^n \delta_{\mathbf{z},\mathbf{0}}.
    \label{eq:bitstringorthog}
\end{equation}
That is, after taking the sum over all bitstrings $\mathbf{x}$, the only term that survives is the $\mathbf{z} = \mathbf{0}$ term. We therefore find
\begin{equation}
	\sum_{\mathbf{x}} \mathbb{E}[q_U(\mathbf{x})] = 1
    \label{eq:k1sumbitstrings}
\end{equation}
as expected for a probability distribution.

\subsubsection{Path Integral}

Although for $k=1$ we are able to solve for the spectrum of $H_{\mathrm{eff}}$ exactly, it will be useful to also solve the problem using mean-field (large-$n$) methods, which will be the only tool available for larger $k > 1$. We start with the expectation value written in terms of the effective Hamiltonian:
\begin{equation}
    \mathbb{E}[q_U(\mathbf{x})] = (-1)^{\mathbf{x}} \bra{\mathbf{x\overline{x}}} e^{-H_{\mathrm{eff}} T} \ket{\mathbf{01}}
\end{equation}
where we subsequently ignore all terms in the Hamiltonian that are subleading in $n$:
\begin{equation}
    H_{\mathrm{eff}} \approx - \frac{J}{2 n} \left( \sum_i \vec{\sigma}_{iL} \cdot \vec{\sigma}_{iR} \right)^2 + \frac{9 J}{2} n + \mathcal{O}(1).
    \label{eq:k1effhamlargen}
\end{equation}
We convert this expression into a path integral by introducing the overcomplete basis of coherent spin states $|\vec{\Omega}_{ia} \rangle$ parameterized by the SO(3) unit vectors $\vec{\Omega}_{ia}$, where $a = L,R$. The coherent states obey the eigenvalue equation
\begin{equation}
    \vec{\Omega}_{ia} \cdot \vec{\sigma}_{ia} |\vec{\Omega}_{ia}\rangle = |\vec{\Omega}_{ia}\rangle
\end{equation}
and have the completeness relation
\begin{equation}
    \mathbb{I} = \frac{1}{2 \pi} \int d^2 \vec{\Omega}_{ia} | \vec{\Omega}_{ia} \rangle \langle \vec{\Omega}_{ia} |.
\end{equation}
We insert this resolution of the identity at each timestep of the propagator $e^{-H_{\mathrm{eff}} T}$ in the usual way to obtain a path integral expression \cite{altland2010condensed}. To turn the spins $\vec{\sigma}_{ia}$ into coherent state vectors, we make use of the `lower symbol' \cite{klauder1985coherent}
% \begin{equation}
%     \sigma^{\alpha}_{ia} = \frac{3}{2 \pi} \int d^2 \vec{\Omega}_{ia} | \vec{\Omega}_{ia} \rangle \langle \vec{\Omega}_{ia} | \ \Omega^{\alpha}_{ia}.
% \end{equation}
\begin{equation}
    \langle \vec{\Omega}_{ia}| \sigma^{\alpha}_{ia} | \vec{\Omega}_{ia} \rangle = \Omega_{ia}^{\alpha}
\end{equation}
This effectively converts the Pauli spin operators into classical SO(3) unit vectors $\vec{\sigma}_{ia} \rightarrow \vec{\Omega}_{ia}$. The path integral expression is
\begin{align}
    (-1)^{\mathbf{x}} \bra{\mathbf{x\overline{x}}} e^{-H_{\mathrm{eff}} T} \ket{\mathbf{01}} = \int^{\mathbf{x\overline{x}}}_{\mathbf{01}} \mathcal{D} \vec{\Omega}_{ia} \ e^{-I[\vec{\Omega}_{ia}]}
\end{align}
with action (to leading order in $n$)
\begin{equation}
    I[\vec{\Omega}_{ia}] = \int_0^T d t \left[ - \frac{J}{2} n \left( \frac{1}{n} \sum_i \vec{\Omega}_{iL} \cdot \vec{\Omega}_{iR} \right)^2 + \frac{9 J}{2} n + \mathcal{O}(1) \right] - \ln (-1)^{\mathbf{x}}
\end{equation}
where we have buried the overlap terms $\langle \vec{\Omega}_{ia,t+1} | \vec{\Omega}_{ia,t} \rangle$ in the definition of the integration measure $\mathcal{D} \vec{\Omega}_{ia}$.

We now introduce the mean field
\begin{equation}
    G(t) = \frac{1}{n} \sum_i \vec{\Omega}_{iL} \cdot \vec{\Omega}_{iR}
\end{equation}
and its associated Lagrange multiplier $F(t)$ via the identity
\begin{equation}
    1 = \int \mathcal{D} F \mathcal{D} G \ \exp \left[ n \int_0^T dt \ F(t) \left( G(t) - \frac{1}{n} \sum_i \vec{\Omega}_{iL} \cdot \vec{\Omega}_{iR} \right) \right]
\end{equation}
to write the path integral as
\begin{equation}
    (-1)^{\mathbf{x}} \bra{\mathbf{x\overline{x}}} e^{-H_{\mathrm{eff}} T} \ket{\mathbf{01}} = \int^{\mathbf{x\overline{x}}}_{\mathbf{01}} \mathcal{D} \vec{\Omega}_{ia} \int \mathcal{D} F \mathcal{D} G \ e^{-I[\vec{\Omega}_{ia},F,G]}
\end{equation}
with action
\begin{equation}
    I[\vec{\Omega}_{ia},F,G] = n \int_0^T dt \left[ - \frac{J}{2} G^2 + \frac{9 J}{2} - F G + \frac{1}{n} \sum_i F \ \vec{\Omega}_{iL} \cdot \vec{\Omega}_{iR} \right] - \ln (-1)^{\mathbf{x}}.
    \label{eq:k1pathintegralaction}
\end{equation}
Next, we integrate over the spin part of the action, which has been decoupled by the fields $G,F$. It is convenient to transform this back into an operator representation, yielding:
\begin{align}
    &(-1)^{\mathbf{x}} \int^{\mathbf{x\overline{x}}}_{\mathbf{01}} \mathcal{D} \vec{\Omega}_{ia} \exp \left( - \int_0^T dt \sum_i F(t) \ \vec{\Omega}_{iL} \cdot \vec{\Omega}_{iR} \right) \nonumber \\
    & \quad \quad \quad \quad \quad \quad \quad \quad = (-1)^{\mathbf{x}} \bra{\mathbf{x\overline{x}}} \mathbb{T} \exp \left[ - \int_0^T dt \ F(t) \sum_i \vec{\sigma}_{iL} \cdot \vec{\sigma}_{iR} \right] \ket{\mathbf{01}} \nonumber \\
        & \quad \quad \quad \quad \quad \quad \quad \quad = \prod_i (-1)^{x_i} \bra{x_i\overline{x}_i} \mathbb{T} \exp \left[ - \int_0^T dt \ F(t) \vec{\sigma}_{iL} \cdot \vec{\sigma}_{iR} \right] \ket{01} 
\end{align}
where $\mathbb{T}$ is the time-ordering operation. Our final path integral expression for the expectation value of the return probability is therefore
\begin{equation}
    \mathbb{E}[q_U(\mathbf{x})] = (-1)^{\mathbf{x}} \bra{\mathbf{x\overline{x}}} e^{-H_{\mathrm{eff}} T} \ket{\mathbf{01}} = \int \mathcal{D} F \mathcal{D} G \ e^{-I[F,G]}
\end{equation}
with action
\begin{equation}
    I[F,G] = - \sum_i \ln \left( (-1)^{x_i} \bra{x_i \overline{x}_i} \mathbb{T} \exp \left[ - \int_0^T dt \ F(t) \vec{\sigma}_{iL} \cdot \vec{\sigma}_{iR} \right] \ket{01} \right) + n \int_0^T dt \left[ - \frac{J}{2} G^2 + \frac{9 J}{2} - F G \right].
\end{equation}

\subsubsection{Large-$n$ Limit}

Next, we solve for the saddle points of this action in the limit of large $n$. This relies on the fact that the action is proportional to $n$, so that fluctuations are suppressed by $1/n$ (this is equivalent to a semiclassical expansion, where $1/n$ plays the role of $\hbar$). We assume that the saddle-points $F,G$ are time-independent, so that the action takes the form
\begin{equation}
    I[F,G] = - \sum_i \ln \left( (-1)^{x_i} \bra{x_i \overline{x}_i} \exp \left[ - F T \vec{\sigma}_{iL} \cdot \vec{\sigma}_{iR} \right] \ket{01} \right) + n \left[ - \frac{JT}{2} G^2 + \frac{9 JT}{2} - F T G \right].
\end{equation}
We can exactly evaluate the first term since it is just a Heisenberg coupling:
\begin{align}
    (-1)^{x_i} \bra{x_i \overline{x}_i} \exp \left[ - F T \vec{\sigma}_{iL} \cdot \vec{\sigma}_{iR} \right] \ket{01} &= (-1)^{x_i} \bra{x_i \overline{x}_i} \left( e^{3 F T} \ket{s} \bra{s}_i + e^{-F T} \ket{t} \bra{t}_i \right) \ket{01} \nonumber \\
    &= \frac{1}{2} e^{3 F T} + (-1)^{x_i} \frac{1}{2} e^{-F T}.
    \label{eq:simplifyheis}
\end{align}
Hence our action is:
\begin{equation}
    I[F,G] = - \sum_i \ln \left[ \frac{1}{2} e^{3 F T} + (-1)^{x_i} \frac{1}{2} e^{-F T} \right] + n \left[ - \frac{JT}{2} G^2 + \frac{9 JT}{2} - F T G \right]
\end{equation}
We now look for saddle points that solve the Euler-Lagrange equations of motion $\partial I / \partial G = 0$, $\partial I / \partial F = 0$. The equation of motion for $G$ gives
\begin{equation}
    F = - J G
\end{equation}
which we substitute back into the action to obtain:
\begin{equation}
    I^* = - \sum_i \ln \left[ \frac{1}{2} e^{-3 G J T} + (-1)^{x_i} \frac{1}{2} e^{G J T} \right] + n \left[ \frac{JT}{2} G^2 + \frac{9 JT}{2} \right].
    \label{eq:spactiong}
\end{equation}
Solving the second equation of motion, we find the saddle point equation of motion for $G$:
\begin{equation}
    G = \frac{1}{n} \sum_i \frac{-3 + (-1)^{x_i} e^{4 G J T}}{1+ (-1)^{x_i} e^{4 G J T}}
\end{equation}
which at long times takes the time-independent value
\begin{equation}
    G \rightarrow  -3.
\end{equation}
Plugging this back into the saddle-point action, we find:
\begin{align}
    \mathbb{E}[q_U(\mathbf{x})] = e^{-I^*} &= \frac{1}{2^n} \prod_i \left(1 + (-1)^{x_i} e^{-12 J T} \right) \nonumber \\
    &= \frac{1}{2^n} \prod_i \sum_{z_i = 0,1} (-1)^{x_i z_i} e^{-12 J T z_i} \nonumber \\
    &= \frac{1}{2^n} \sum_{\mathbf{z} = \{0,1\}} (-1)^{\mathbf{x} \cdot \mathbf{z}} \exp \left( -12 J T \magn{\mathbf{z}} \right)
    \label{eq:k1saddlepoint}
\end{align}
which is equivalent to Eq. \eqref{eq:k1final} in the limit that $n \gg \magn{\mathbf{z}}$. Notice, starting from the first line, that we may alternatively write this expression as
\begin{equation}
    \mathbb{E}[q_U(\mathbf{x})] = \frac{1}{2^n} \left(1 - e^{-12 J T} \right)^{\magn{\mathbf{x}}} \left(1 + e^{-12 J T} \right)^{n - \magn{\mathbf{x}}}
\end{equation}
so that the sum over bitstrings $\mathbf{x}$ is trivial:
\begin{equation}
    \sum_{\mathbf{x}} \mathbb{E}[q_U(\mathbf{x})] = 1
\end{equation}

\subsection{Higher Moments}

Next, we consider higher moments $\mathbb{E}[q_U^k(\mathbf{x})]$ of the Brownian transition probability, with $k > 1$. In this case there are $2k$ replicas, which we label using the combined index $ra = 1L,1R,2L,2R, \ldots, kL,kR$, where $r = 1,2,\ldots,k$ and $a = L,R$. For example, for $k = 2$ we have:
\begin{align}
    \mathbb{E}[q_U^2(\mathbf{x})] &= \bra{\mathbf{x\overline{x} x\overline{x}}} \mathbb{E}[ U \otimes U^{\mathcal{T}} \otimes U \otimes U^{\mathcal{T}}] \ket{\mathbf{0101}} \nonumber \\
    &= \bra{\mathbf{x\overline{x} x\overline{x}}} e^{-H^{(2)}_{\mathrm{eff}} T} \ket{\mathbf{0101}}
\end{align}
where there are two factors of $(-1)^{\mathbf{x}}$ that cancel each other. For higher moments $k > 2$ we similarly have
\begin{align}
    \mathbb{E}[q_U^k(\mathbf{x})] &= \left[(-1)^{\mathbf{x}}\right]^k \bra{\mathbf{x\overline{x} x\overline{x}} \cdots} \mathbb{E}[ U \otimes U^{\mathcal{T}} \otimes U \otimes U^{\mathcal{T}} \cdots] \ket{\mathbf{0101} \cdots} \nonumber \\
    &= \left[(-1)^{\mathbf{x}}\right]^k \bra{\mathbf{x\overline{x} x\overline{x}} \cdots} e^{-H^{(k)}_{\mathrm{eff}} T} \ket{\mathbf{0101} \cdots}.
\end{align}

\subsubsection{Effective Hamiltonian}

Performing the disorder average over the Brownian couplings $J_{ij}^{\alpha \beta}(t)$ as above, we find the effective Hamiltonian for arbitrary $k$:
\begin{equation}
    H^{(k)}_{\mathrm{eff}} = \frac{J}{2N} \sum_{\substack{i<j \\ \alpha \beta}} \left( \sum_{r = 1}^k \sigma_{irL}^{\alpha} \sigma_{jrL}^{\beta} - \sigma_{irR}^{\alpha} \sigma_{jrR}^{\beta} \right)^2
    \label{eq:arbitrarykeffham}
\end{equation}
generalizing Eq. \eqref{eq:k1effectiveham}, where we use the convention $(-1)^L = 1, (-1)^R = -1$. After some manipulation, this Hamiltonian can be rewritten as
\begin{equation}
    H^{(k)}_{\mathrm{eff}} = \frac{J}{2 n} \sum_{ra < sb} (-1)^{a+b} \left( \sum_i \vec{\sigma}_{ira} \cdot \vec{\sigma}_{isb} \right)^2 + \frac{9}{2} k J (n-1) - \frac{J}{2 n} \sum_{ra < sb} (-1)^{a+b} \sum_i \left( \vec{\sigma}_{ira} \cdot \vec{\sigma}_{isb} \right)^2
\end{equation}
where in the sum over $ra < sb$ we treat $ra$ and $sb$ as a combined index, with the convention that $L < R$, such that there are a total of $k (2k-1)$ terms in the sum. The last term in the effective Hamiltonian is sub-leading at large $n$, so we subsequently ignore it. We can see that this Hamiltonian has exactly $k!$ ground states, given by different pairings of $L,R$ replicas. 
For example, consider the state
\begin{equation}
    \ket{\Omega} = \prod_{ir} \ket{s}_{irL,irR}
\end{equation}
which is a product state of singlets between spins $\vec{\sigma}_{irL}$ and $\vec{\sigma}_{irR}$ for each $r = 1,2,\ldots,k$. We draw this state using the `ladder' diagram shown in Fig. \ref{fig:groundstatediagrams}.a. This state has the property that we can `pass' the Pauli operators through the singlet state from $L$ to $R$ with a minus sign:
\begin{equation}
    \sigma^{\alpha}_{irL} \ket{\Omega} = -\sigma^{\alpha}_{irR} \ket{\Omega}
\end{equation}
and therefore
\begin{equation}
    \sigma_{irL}^{\alpha} \sigma_{jrL}^{\beta} \ket{\Omega} = \sigma_{irR}^{\alpha} \sigma_{jrR}^{\beta} \ket{\Omega}
    \label{eq:groundstatepassthrough}
\end{equation}
so that the two terms of opposite sign in the Hamiltonian Eq. \eqref{eq:arbitrarykeffham} cancel, leading to $H^{(k)}_{\mathrm{eff}}\ket{\Omega} = 0$. This cancellation property works for any pairing of $L,R$ spins, so we find exactly $k!$ ground states, all having energy $0$ as illustrated in Fig. \ref{fig:groundstatediagrams}.

\begin{figure}
    \centering
    \includegraphics{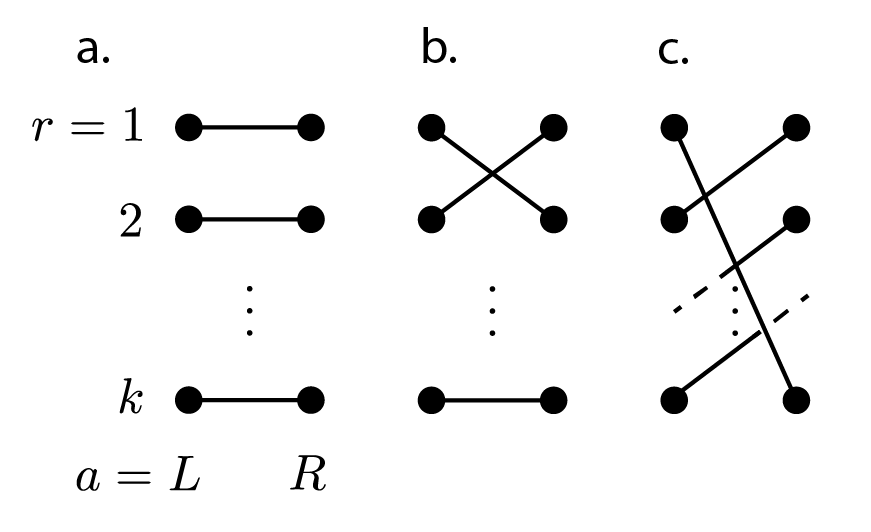}
    \caption{A sample of possible ground-state pairings for the effective Hamiltonian $H^{(k)}_{\mathrm{eff}}$. Lines represent singlet states $\ket{s}$ connecting pairs of spins in the $L,R$ replicas. (a) `Ladder' diagram showing singlets between $L,R$ spins sharing the same $r$ index. (b) Permuted diagram where the $r = 1,2$ spins have been exchanged on the $L$ side. (c) Cyclically permuted diagram where all spins on the $L$ side have been cyclically permuted. All $k!$ ground states can be obtained by starting from the `ladder' diagram and permuting spins on the $L$ side.}
    \label{fig:groundstatediagrams}
\end{figure}

\subsubsection{Path Integral}

Using the same tools as above, we convert this expectation value into a path integral so that we may use large-$n$ methods to evaluate it. In doing so we introduce $k (2k-1)$ mean fields
\begin{equation}
    G_{rs}^{ab}(t) = \frac{1}{n} \sum_i \vec{\Omega}_{ira} \cdot \vec{\Omega}_{isb}
\end{equation}
for $ra < sb$, and their associated Lagrange multipliers $F_{rs}^{ab}$(t) via the identity
\begin{equation}
    1 = \int \mathcal{D} F_{rs}^{ab} \mathcal{D} G_{rs}^{ab} \ \exp \left[ n \int_0^T dt \ \sum_{ra < sb} F_{rs}^{ab}(t) \left( G_{rs}^{ab}(t) - \frac{1}{n} \sum_i \vec{\Omega}_{ira} \cdot \vec{\Omega}_{isb} \right) \right].
\end{equation}
Our path integral expression for the expectation value of the $k$th moment, generalizing Eq. \eqref{eq:k1pathintegralaction}, is
\begin{equation}
    \left[(-1)^{\mathbf{x}}\right]^k \bra{\mathbf{x\overline{x}} \cdots} e^{-H^{(k)}_{\mathrm{eff}} T} \ket{\mathbf{01} \cdots} = \int^{\mathbf{x\overline{x}}\cdots}_{\mathbf{01} \cdots} \mathcal{D} \vec{\Omega}_{ira} \int \mathcal{D} F_{rs}^{ab} \mathcal{D} G_{rs}^{ab} \ e^{-I[\vec{\Omega}_{ira},F_{rs}^{ab},G_{rs}^{ab}]}
\end{equation}
with action
\begin{align}
    I[\vec{\Omega}_{ira},F_{rs}^{ab},G_{rs}^{ab}] &= n \int_0^T dt \left[ \sum_{ra < sb} \left( (-1)^{a+b} \frac{J}{2} \left(G_{rs}^{ab} \right) ^2 - F_{rs}^{ab} G_{rs}^{ab} \right. \right. \nonumber \\
    &\quad \quad \quad \quad \quad \quad \quad \quad \quad \left. \left. + \frac{1}{n} \sum_i F_{rs}^{ab} \ \vec{\Omega}_{ira} \cdot \vec{\Omega}_{isb} \right) + \frac{9 k J}{2} \right] - \ln \left[(-1)^{\mathbf{x}}\right]^k.
\end{align}
Next, we integrate over the spin part of the action, which has been decoupled by the fields $G_{rs}^{ab},F_{rs}^{ab}$. It is convenient to transform this back into an operator representation, yielding:
\begin{align}
    &\left[(-1)^{\mathbf{x}}\right]^k \int^{\mathbf{x\overline{x}} \cdots}_{\mathbf{01} \cdots} \mathcal{D} \vec{\Omega}_{ira} \exp \left( - \int_0^T dt \sum_{ra < sb} \sum_i F_{rs}^{ab}(t) \ \vec{\Omega}_{ira} \cdot \vec{\Omega}_{isb} \right) \nonumber \\
    & \quad \quad \quad \quad \quad \quad \quad \quad = \left[(-1)^{\mathbf{x}} \right]^k \bra{\mathbf{x\overline{x}} \cdots} \mathbb{T} \exp \left[ - \int_0^T dt \sum_{ra < sb} \ F_{rs}^{ab}(t) \sum_i \vec{\sigma}_{ira} \cdot \vec{\sigma}_{isb} \right] \ket{\mathbf{01} \cdots} \nonumber \\
        & \quad \quad \quad \quad \quad \quad \quad \quad = \prod_i \left[(-1)^{x_i}\right]^k \bra{x_i\overline{x}_i \cdots} \mathbb{T} \exp \left[ - \int_0^T dt \sum_{ra < sb} \ F_{rs}^{ab}(t) \vec{\sigma}_{ira} \cdot \vec{\sigma}_{isb} \right] \ket{01 \cdots} 
\end{align}
where $\mathbb{T}$ is the time-ordering operation. Our final path integral expression for the expectation value of the return probability is therefore
\begin{equation}
    \mathbb{E}[q_U^k(\mathbf{x})] = \int \mathcal{D} F_{rs}^{ab} \mathcal{D} G_{rs}^{ab} \ e^{-I[F_{rs}^{ab},G_{rs}^{ab}]}
\end{equation}
with action
\begin{align}
    I[F_{rs}^{ab},G_{rs}^{ab}] &= - \sum_i \ln \left( \left[(-1)^{x_i}\right]^k \bra{x_i \overline{x}_i \cdots} \mathbb{T} \exp \left[ - \int_0^T dt \ \sum_{ra < sb} F_{rs}^{ab}(t) \vec{\sigma}_{ira} \cdot \vec{\sigma}_{isb} \right] \ket{01 \cdots} \right) \nonumber \\
    & \quad \quad \quad \quad + n \int_0^T dt \sum_{ra < sb} \left[ (-1)^{a+b} \frac{J}{2} \left( G_{rs}^{ab} \right)^2 + \frac{9 k J}{2} - F_{rs}^{ab} G_{rs}^{ab} \right].
\end{align}

\subsubsection{Large-$n$ Limit}

Next, we solve for the saddle points of this action in the limit of large $n$. We again assume that the saddle-points $F_{rs}^{ab},G_{rs}^{ab}$ are time-independent, so that the action takes the form
\begin{align}
    I[F_{rs}^{ab},G_{rs}^{ab}] &= - \sum_i \ln \left( \left[(-1)^{x_i}\right]^k \bra{x_i \overline{x}_i \cdots} \exp \left[ - \sum_{ra < sb} F_{rs}^{ab} T \ \vec{\sigma}_{ira} \cdot \vec{\sigma}_{isb} \right] \ket{01 \cdots} \right) \nonumber \\
    & \quad \quad \quad \quad + n \sum_{ra < sb} \left[ (-1)^{a+b} \frac{JT}{2} \left( G_{rs}^{ab} \right)^2 + \frac{9 k JT}{2} - F_{rs}^{ab} G_{rs}^{ab} T \right].
\end{align}
Solving the Euler-Lagrange equation $\partial I / \partial G_{rs}^{ab} = 0$ we find
\begin{equation}
    F_{rs}^{ab} = (-1)^{a+b} J G_{rs}^{ab}
\end{equation}
which we substitute back into the action to find
\begin{align}
    I^* &= - \sum_i \ln \left( \left[(-1)^{x_i}\right]^k \bra{x_i \overline{x}_i \cdots} \exp \left[ - JT \sum_{ra < sb} (-1)^{a+b} G_{rs}^{ab} \ \vec{\sigma}_{ira} \cdot \vec{\sigma}_{isb} \right] \ket{01 \cdots} \right) \nonumber \\
    & \quad \quad \quad \quad + n \sum_{ra < sb} \left[ -(-1)^{a+b} \frac{JT}{2} \left( G_{rs}^{ab} \right)^2 + \frac{9 k JT}{2} \right].
\end{align}
Inspired by the ground-state pairings that we found above, we make the ansatz $G_{rs}^{ab} = 0$ unless $r = s$ and $ab = LR$, which should correspond to the `ladder' diagram in Fig. \ref{fig:groundstatediagrams}.a. We show that this ansatz leads to a self-consistent saddle-point. In this case, the action simplifies to a sum over independent actions
\begin{align}
    I^* &= - \sum_i \ln \left( \left[(-1)^{x_i}\right]^k \bra{x_i \overline{x}_i \cdots} \exp \left[ JT \sum_{r} G_{rr}^{LR} \ \vec{\sigma}_{irL} \cdot \vec{\sigma}_{irR} \right] \ket{01 \cdots} \right) \nonumber \\
    & \quad \quad \quad \quad + n \sum_{r} \left[ \frac{JT}{2} \left( G_{rr}^{LR} \right)^2 + \frac{9 k JT}{2} \right]. \nonumber \\
    &= \sum_r I^*_r
\end{align}
where
\begin{align}
    I^*_r &= \sum_i \ln \left( (-1)^{x_i} \bra{x_i \overline{x}_i} \exp \left[ JT G_{rr}^{LR} \ \vec{\sigma}_{irL} \cdot \vec{\sigma}_{irR} \right] \ket{01} \right) \nonumber \\
    & \quad \quad \quad \quad + n \left[ \frac{JT}{2} \left( G_{rr}^{LR} \right)^2 + \frac{9 JT}{2} \right].
\end{align}
We may evaluate each of these actions independently.
Similar to Eq. \eqref{eq:simplifyheis}, we can exactly evaluate the first term since it is just a Heisenberg coupling:
\begin{align}
    &(-1)^{x_i} \bra{x_i \overline{x}_i} \exp \left[ JT G_{rr}^{LR} \vec{\sigma}_{irL} \cdot \vec{\sigma}_{irR} \right] \ket{01} \nonumber \\
    & \quad \quad \quad \quad = (-1)^{x_i} \bra{x_i \overline{x}_i} \left( e^{-3 JT G_{rr}^{LR}} \ket{s} \bra{s}_i + e^{JT G_{rr}^{LR}} \ket{t} \bra{t}_i \right) \ket{01} \nonumber \\
    &= \frac{1}{2} e^{-3 JT G_{rr}^{LR}} + (-1)^{x_i} \frac{1}{2} e^{ JT G_{rr}^{LR}}.
\end{align}
This gives the action
\begin{equation}
    I_r^* = - \sum_i \ln \left[ \frac{1}{2} e^{-3 G_{rr}^{LR} J T} + (-1)^{x_i} \frac{1}{2} e^{G_{rr}^{LR} J T} \right] + n \left[ \frac{JT}{2} \left( G_{rr}^{LR} \right)^2 + \frac{9 JT}{2} \right].
\end{equation}
which is equivalent to Eq. \eqref{eq:spactiong} with $G \rightarrow G_{rr}^{LR}$. We therefore find the same time-independent saddle point $G_{rr}^{LR} \rightarrow -3$, leading to 
\begin{equation}
    e^{-I^*} = \prod_{r=1}^k e^{-I^*_r} = \left( \frac{1}{2^n} \sum_{\mathbf{z} = \{0,1\}} (-1)^{\mathbf{x} \cdot \mathbf{z}} \exp \left( -12 J T \magn{\mathbf{z}} \right) \right)^k.
\end{equation}
This is the saddle-point corresponding to the `ladder' diagram, but we also have many other saddle points corresponding to permutations of the $L$ replicas. There are $k!$ such saddle points, all with the same action, so we conclude that
\begin{align}
    \mathbb{E}[q_U^k(\mathbf{x})] &= k! \ e^{-I^*} \nonumber \\
    &= k! \left( \frac{1}{2^n} \sum_{\mathbf{z} = \{0,1\}} (-1)^{\mathbf{x} \cdot \mathbf{z}} \exp \left( -12 J T \magn{\mathbf{z}} \right) \right)^k.
\end{align}
Notice that for $\mathbf{x} = \mathbf{0}$, this expression simplifies to
\begin{equation}
    \mathbb{E}[q_U^k(\mathbf{0})] = \frac{k!}{2^{kN}} \left( 1 + e^{-12 J T} \right)^{kN}.
\end{equation}
We also find that as $T \rightarrow 0$, Eq. \eqref{eq:arbitrarykfinal} vanishes unless $\mathbf{x} = 0$, which makes sense. Finally, for $T \rightarrow \infty$, this expression limits to the Porter-Thomas value $\mathbb{E}[q_U^k(\mathbf{x})] \rightarrow k! \ 2^{-kN}$. We can further simplify the above expression using the following manipulations:
\begin{align}
    \sum_{\mathbf{z} = \{0,1\}} (-1)^{\mathbf{x} \cdot \mathbf{z}} \exp \left( -12 J T \magn{\mathbf{z}} \right) &= \frac{1}{2^n} \prod_i \sum_{z_i = 0,1} (-1)^{x_i z_i} e^{-12 J T z_i} \nonumber \\
    &= \prod_i \left(1 + (-1)^{x_i} e^{-12 J T} \right) \nonumber \\
    &= \left(1 - e^{-12 J T} \right)^{\magn{\mathbf{x}}} \left(1 + e^{-12 J T} \right)^{n-\magn{\mathbf{x}}} 
\end{align}
Thus we find our final large-$n$ expression for the $k$th moment:
\begin{equation}
    \mathbb{E}[q_U^k(\mathbf{x})] = \frac{k!}{2^{kN}} \left(1 - e^{-12 J T} \right)^{k \magn{\mathbf{x}}} \left(1 + e^{-12 J T} \right)^{k (n-\magn{\mathbf{x}})} 
    \label{eq:arbitrarykfinal}
\end{equation}
Finally, we perform the sum over bitstrings $\mathbf{x}$ to obtain
\begin{equation}
    \sum_{\mathbf{x}} \mathbb{E}[q_U^k(\mathbf{x})] = \frac{k!}{2^{kN}} \left[ \left(1 - e^{-12 J T} \right)^{k} +  \left(1 + e^{-12 J T} \right)^{k} \right]^n
    \label{eq:arbitraryksumbitstrings}
\end{equation}

\subsubsection{Second Moment $k = 2$}

We may simplify Eq. \eqref{eq:arbitrarykfinal} for $k = 2$. In this case we have
\begin{equation}
    \mathbb{E}[q_U^2(\mathbf{x})] = \frac{2}{2^{2N}}\sum_{\mathbf{z},\mathbf{z'} = \{0,1\}} (-1)^{\mathbf{x} \cdot (\mathbf{z}+\mathbf{z'})} \exp \left( -12 J T (\magn{\mathbf{z}}+\magn{\mathbf{z'}}) \right)
\end{equation}
Taking the sum over bitstrings $\mathbf{x}$ and applying the identity Eq. \eqref{eq:bitstringorthog}, we find that all terms vanish unless $\mathbf{z} = \mathbf{z'}$, which gives:
\begin{align}
    \sum_{\mathbf{x}} \mathbb{E}[q_U^2(\mathbf{x})] &= \frac{2}{2^{n}}\sum_{\mathbf{z} = \{0,1\}} \exp \left( -24 J T \magn{\mathbf{z}} \right) \nonumber \\
    &= \frac{2}{2^{n}} \left( 1 + e^{-24 J T} \right)^n.
    \label{eq:k2sumbitstrings}
\end{align}
Therefore, the true quantum circuit gives an \textsf{XEB} score of
\begin{equation}
    \mathrm{XEB}(U,U) \approx 1 + 2 n e^{-24 J T} + \mathcal{O}(e^{-48 J T})
\end{equation}

\subsection{The spoofing algorithm of \cite{gao2024limitations}}

In this section, we consider a classically efficient spoofing algorithm \cite{gao2024limitations} that aims to generate bitstrings drawn from a probability distribution $A_U(\mathbf{x})$ that spoofs the distribution $q_U(\mathbf{x})$ generated by the full Brownian circuit in terms of the linear cross-entropy benchmark (\textsf{XEB}) score:
\begin{equation}
    \mathrm{XEB}(U,A) = 2^n \sum_{\mathbf{x}} \mathbb{E} [ q_U(\mathbf{x}) A_U(\mathbf{x}) ] - 1.
\end{equation}
% First, let's consider some simple examples. If we use the real distribution $q_U(\mathbf{x})$, then from Eq. \eqref{eq:k2sumbitstrings} at long times $T \rightarrow \infty$ we obtain a score of $\mathrm{XEB}(U,U) = 1$. If we use a flat distribution $q(\mathbf{x}) = 2^{-n}$ that samples bitstrings uniformly, then we obtain a score of $\mathrm{XEB}(U,A) = 0$. If we use a delta-function distribution $q(\mathbf{x}) = \delta_{\mathbf{x},\mathbf{0}}$ then from Eq. \eqref{eq:p0expecval} we obtain $\mathrm{XEB}(U,A) \rightarrow 2^n - 1$ as $T \rightarrow 0$ and $\mathrm{XEB}(U,A) \rightarrow 0$ as $T \rightarrow \infty$. \gsb{Is this a problem?}

To spoof a high \textsf{XEB} score at late times, the authors of \cite{gao2024limitations} suggested splitting up the original circuit into disjoint subsystems, each of which could be efficiently simulated classically. Here we apply this idea to our all-to-all Brownian circuits by splitting the original circuit into $K = n/\log n$ disjoint subsystems of $M = \log n$ qubits each. Unlike the nearest-neighbor circuits considered in \cite{gao2024limitations}, this splitting requires deleting a huge number of couplings ($\mathcal{O}(n^2)$), yielding very slow dynamics in each of the subsystems. To compensate for this, we increase the strength of the couplings in each subsystem by a factor of $A = \sqrt{K}$ such that each subsystem reaches its Haar-random distribution on the same timescale as the original circuit does.

\begin{figure}
    \centering
    \includegraphics{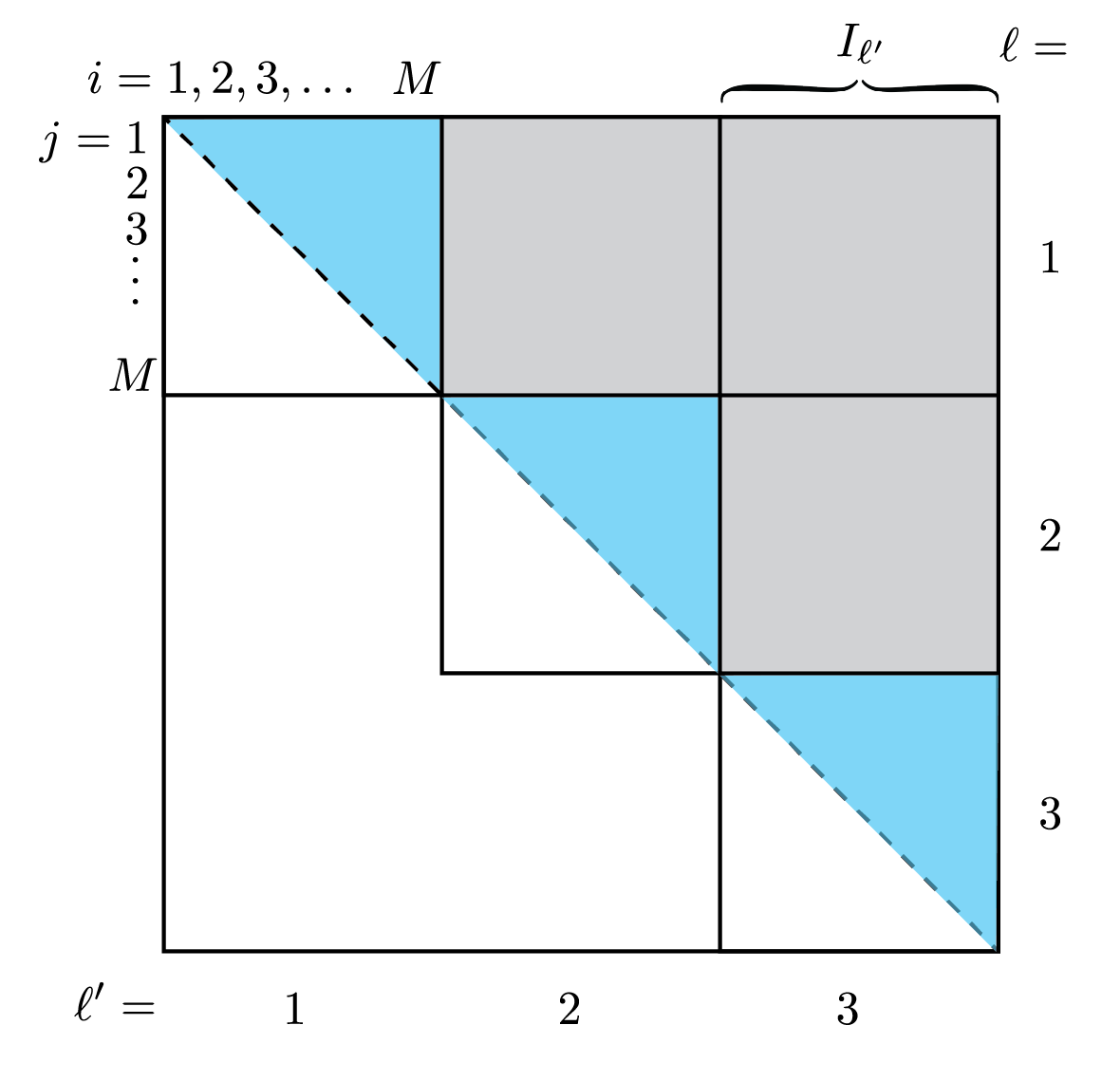}
    \caption{Two categories of coupling coefficients $J_{ij}^{\alpha \beta}(t)$ in the disjoint Brownian circuit model. Blue coefficients couple subsystems to themselves while gray coefficients couple two different subsystems. Gray coefficients are set to zero in the disjoint Brownian model.}
    \label{fig:disjointcouplings}
\end{figure}

To be specific, the original circuit $U$ features $\frac{9}{2} n(n-1)$ Brownian coefficients $J_{ij}^{\alpha \beta}(t)$ for $i < j$. The disjoint circuit $V = \otimes_{\ell} V_{\ell}$ is split into $K$ disjoint subsystems labelled by $\ell = 1,2,\ldots,K$. Each subsystem $V_{\ell}$ consists of $M$ qubits indexed by $i \in I_{\ell}$, where the interval $I_{\ell}$ consists of exactly $M$ indices (we do not bother to enumerate these indices here). The $\frac{9}{2} n(n-1)$ coefficients $J_{ij}^{\alpha \beta}(t)$ in the original circuit can then be divided into categories depending on whether they couple a subsystem $\ell$ to itself or whether they couple a subsystem $\ell$ to a different subsystem $\ell'$, which we illustrate in Fig. \ref{fig:disjointcouplings}. There are $\frac{9}{2} K M (M-1)$ coefficients that couple subsystems to themselves (blue in Fig. \ref{fig:disjointcouplings}), where $i < j \in I_{\ell}$ for any $\ell$. In addition, there are $\frac{9}{2} M^2 K(K-1)$ coefficients that couple different subsystems (gray in Fig. \ref{fig:disjointcouplings}), where $i \in I_{\ell}$ and $j \in I_{\ell'}$ for $\ell \neq \ell'$. In the disjoint circuit $V$, this second category of coefficients are all set to zero. Together, these two categories add up to the total $\frac{9}{2} n (n-1)$ for $n = K M$.

\subsubsection{Effective Hamiltonian}

We are interested in calculating the expected overlap
\begin{equation}
    \mathbb{E}[q_U(\mathbf{x}) A_U(\mathbf{x})] = \mathbb{E} \left[ \magn{\bra{\mathbf{x}} U \ket{\mathbf{0}}}^2 \magn{\bra{\mathbf{x}} V \ket{\mathbf{0}}}^2 \right]
\end{equation}
between the original circuit $U$ and the disjoint circuit $V$. Introducing four replicas and applying the time-reversal operation as usual, we find
\begin{align}
    \mathbb{E}[q_U(\mathbf{x}) A_U(\mathbf{x})] &= \bra{\mathbf{x\overline{x} x\overline{x}}} \mathbb{E}[ U \otimes U^{\mathcal{T}} \otimes V \otimes V^{\mathcal{T}}] \ket{\mathbf{0101}} \nonumber \\
    &= \bra{\mathbf{x\overline{x} x\overline{x}}} \prod_t \mathbb{E}[ U_t \otimes U_t^{\mathcal{T}} \otimes V_t \otimes V_t^{\mathcal{T}}] \ket{\mathbf{0101}}
\end{align}
where replicas $1R,1L$ refer to the full Brownian circuit $U$ and replicas $2R,2L$ refer to the disjoint Brownian circuit $V$. Similar to the calculations above, the expectation value at each time $t$ can be written as a series of Gaussian integrals:
\begin{align}
    &\mathbb{E}[ U_t \otimes U_t^{\mathcal{T}} \otimes V_t \otimes V_t^{\mathcal{T}}] \nonumber \\
    &= \int \prod_{\substack{i<j \\ \alpha \beta}} dJ_{ij}^{\alpha \beta}(t) \exp \left[ - \sum_{\substack{i<j \\ \alpha\beta}} \frac{\left( J_{ij}^{\alpha \beta}(t)\right)^2}{2J / \Delta t n} \right] \nonumber \\
    &\times \exp \left[ -i \sum_{\ell < m} \sum_{\substack{i \in I_{\ell} \\ j \in I_{m}}} \sum_{\alpha \beta} J_{ij}^{\alpha \beta}(t) \left( \sigma_{i1L}^{\alpha} \sigma_{j1L}^{\beta} - \sigma_{i1R}^{\alpha} \sigma_{j1R}^{\beta} \right) \Delta t \right] \nonumber \\
    &\times \exp \left[ -i \sum_{\ell} \sum_{i < j \in I_{\ell}} \sum_{\alpha \beta} J_{ij}^{\alpha \beta}(t) \left( \sigma_{i1L}^{\alpha} \sigma_{j1L}^{\beta} - \sigma_{i1R}^{\alpha} \sigma_{j1R}^{\beta} + A \sigma_{i2L}^{\alpha} \sigma_{j2L}^{\beta} - A \sigma_{i2R}^{\alpha} \sigma_{j2R}^{\beta} \right) \Delta t \right] \nonumber \\
    &= e^{-H_{\mathrm{eff}} \Delta t}
\end{align}
where the third line accounts for the couplings between different subsystems (gray in Fig. \ref{fig:disjointcouplings}) that are present in the full Brownian circuit $U$ but not the disjoint circuit $V$, while the fourth line accounts for the couplings within each subsystem (blue in Fig. \ref{fig:disjointcouplings}) that are present in both circuits. Notice that we have included an extra factor of $A$ in replicas $2L,2R$ increase the strength of these couplings so that the disjoint system reaches its Haar value on the same timescale as the full circuit. We will find that the choice $A = \sqrt{K}$ gives the right timescale, but we leave this parameter undetermined for the moment. Performing the Gaussian integrals, we find that the effective Hamiltonian is
\begin{align}
    H_{\mathrm{eff}} &= \frac{J}{2 n} \sum_{\ell < m} \sum_{\substack{i \in I_{\ell} \\ j \in I_{m}}} \sum_{\alpha \beta} \left( \sigma_{i1L}^{\alpha} \sigma_{j1L}^{\beta} - \sigma_{i1R}^{\alpha} \sigma_{j1R}^{\beta} \right)^2  \nonumber \\
    &+ \frac{J}{2N} \sum_{\ell} \sum_{i < j \in I_{\ell}} \sum_{\alpha \beta} \left( \sigma_{i1L}^{\alpha} \sigma_{j1L}^{\beta} - \sigma_{i1R}^{\alpha} \sigma_{j1R}^{\beta} + A \sigma_{i2L}^{\alpha} \sigma_{j2L}^{\beta} - A \sigma_{i2R}^{\alpha} \sigma_{j2R}^{\beta} \right)^2 .
    \label{eq:harvardalgham}
\end{align}

Note that this Hamiltonian has a ground state
\begin{equation}
	\ket{\Omega} = \bigotimes_{ir} \ket{s}_{irL,irR}
\end{equation}
which is a `ladder' state consisting of singlet states pairing qubits $\vec{\sigma}_{irL}$ and $\vec{\sigma}_{irR}$ for $r = 1,2$. Similar to Eq. \eqref{eq:groundstatepassthrough}, this state has the property
\begin{equation}
    \sigma_{irL}^{\alpha} \sigma_{jrL}^{\beta} \ket{\Omega} = \sigma_{irR}^{\alpha} \sigma_{jrR}^{\beta} \ket{\Omega}
\end{equation}
so that the terms of opposite sign in the Hamiltonian Eq. \eqref{eq:harvardalgham} cancel, leading to $H_{\mathrm{eff}}\ket{\Omega} = 0$. Since $k = 2$, we can ask whether this Hamiltonian also has the permuted ground state
\begin{equation}
    \ket{\Omega'} = \bigotimes_{i} \ket{s}_{i1L,i2R} \otimes \ket{s}_{i2L,i1R}
\end{equation}
(see Fig. \ref{fig:groundstatediagrams}.b.). However, this is not a ground state for two reasons. First, the two terms in the first line of Eq. \eqref{eq:harvardalgham} do not cancel. Second, assuming $A > 1$, the terms in the second line also do not cancel. So this Hamiltonian only has a single ground state $\ket{\Omega}$.

After some manipulation, the effective Hamiltonian can be written as
\begin{equation}
    H_{\mathrm{eff}} = H_N + H_{A M} + H_{A^2} + H_{A} + H_1.
    \label{eq:harvardeffectivehamtotal}
\end{equation}
Here we have organized the terms by order of magnitude, where
\begin{align}
    H_N &= \frac{J}{2N} \left[ 9 M^2 K(K-1) + 9 (1+A^2) K M(M-1) \right] \nonumber \\
    &- \frac{J}{2N} \left[ \left( \sum_{i=1}^n \heis{i1L}{i1R} \right)^2 + A^2 \sum_{\ell} \left( \sum_{i \in I_{\ell}} \heis{i2L}{i2R} \right)^2 \right]
    \label{eq:harvardhamn}
\end{align}
is $\mathcal{O}(n)$,
\begin{align}
    H_{AM} &= \frac{J}{2N} \left[ A \sum_{\ell} \left( \sum_{i \in I_{\ell}} \heis{i1L}{i2L} \right)^2 -  A \sum_{\ell} \left( \sum_{i \in I_{\ell}} \heis{i1L}{i2R} \right)^2 \right. \nonumber \\
    & \quad \quad \quad \quad \left. - A \sum_{\ell} \left( \sum_{i \in I_{\ell}} \heis{i1R}{i2L} \right)^2 +  A \sum_{\ell} \left( \sum_{i \in I_{\ell}} \heis{i1R}{i2R} \right)^2 \right]
\end{align}
is $\mathcal{O}(AM)$,
\begin{equation}
    H_{A^2} = \frac{J}{2N} A^2 \sum_{\ell} \sum_{i \in I_{\ell}} \left( \heis{i2L}{i2R} \right)^2
\end{equation}
is $\mathcal{O}(A^2)$
\begin{align}
    H_{A} &= - \frac{J}{2N} \left[ A \sum_{\ell} \sum_{i \in I_{\ell}} \left( \heis{i1L}{i2L} \right)^2 -  A \sum_{\ell} \sum_{i \in I_{\ell}} \left( \heis{i1L}{i2R} \right)^2 \right. \nonumber \\
    & \quad \quad \quad \quad \left. - A \sum_{\ell} \sum_{i \in I_{\ell}} \left( \heis{i1R}{i2L} \right)^2 +  A \sum_{\ell} \sum_{i \in I_{\ell}} \left( \heis{i1R}{i2R} \right)^2 \right]
\end{align}
is $\mathcal{O}(A)$, and
\begin{equation}
    H_{1} = \frac{J}{2N} \sum_{\ell} \sum_{i \in I_{\ell}} \left( \heis{i1L}{i1R} \right)^2
\end{equation}
is $\mathcal{O}(1)$. In the limit $n \rightarrow \infty$, most of these terms are negligible compared to the $\mathcal{O}(n)$ terms, so we can write
\begin{equation}
    H_{\mathrm{eff}} \approx H_N
\end{equation}
where
\begin{align}
    H_N &\approx \frac{J}{2N} \left[ 9 n^2 + 9 A^2 K M^2 \right] \nonumber \\
    &- \frac{J}{2N} \left[ \left( \sum_{i=1}^n \heis{i1L}{i1R} \right)^2 + A^2 \sum_{\ell} \left( \sum_{i \in I_{\ell}} \heis{i2L}{i2R} \right)^2 \right]
    \label{eq:harvardhamlargen}
\end{align}
to leading in order in $n$. As a quick check, we can set $\heis{i1L}{i1R} = \heis{i2L}{i2R} = -3$, with all other Heisenberg terms vanishing, and see explicitly that $H_{\mathrm{eff}} = 0$, both for the full Hamiltonian and for the large-$n$ Hamiltonian. We already see from Eq. \eqref{eq:harvardhamlargen} that we must set $A = \sqrt{K}$ in order for the replica $r = 2$ terms to be of the same order as the replica $r = 1$ terms.

\subsubsection{Effective Hamiltonian Spectrum}

Note that the dominant term $H_N$ splits into $K+1$ disjoint pieces:
\begin{equation}
    H_{\mathrm{eff}} \approx H_N = H_0 + \sum_{\ell} H_{\ell}
    \label{eq:harvardh0hl}
\end{equation}
where
\begin{equation}
    H_0 = - \frac{J}{2 n} \left( \sum_{i=1}^n \vec{\sigma}_{i1L} \cdot \vec{\sigma}_{i1R} \right)^2 + \frac{9 J}{2} n
\end{equation}
and 
\begin{equation}
    H_{\ell} = - \frac{J A^2}{2 n} \left( \sum_{i\in I{\ell}} \vec{\sigma}_{i2L} \cdot \vec{\sigma}_{i2R} \right)^2 + \frac{9 J A^2}{2N} M^2.
\end{equation}
These Hamiltonians are both of the same form as Eq. \eqref{eq:k1effhamlargen} and we already know how to find their exact spectra using the methods discussed in Sec. \ref{sec:exactspectrum}. The eigenvectors of $H_0$ are
\begin{equation}
    \ket{\Psi}_0 = \bigotimes_{i=1}^n \ket{\psi}_{i1L,i1R}
\end{equation}
where $\ket{\psi}_{i1L,i1R}$ can independently take one of the values $\ket{s},\ket{t},\ket{00},\ket{11}$. The eigenvalues are
\begin{equation}
    H_0 \ket{\Psi}_0 = \left( \frac{4 J z (3 n - 2z)}{n} \right) \ket{\Psi}_0
\end{equation}
where $z$ counts the number of excitations $\ket{t}, \ket{00}, \ket{11}$. Similarly, the eigenvectors of $H_{\ell}$ are
\begin{equation}
    \ket{\Psi}_{\ell} = \bigotimes_{i \in I_{\ell}} \ket{\psi}_{i2L,i2R}
\end{equation}
where $\ket{\psi}_{i2L,i2R}$ can independently take one of the values $\ket{s},\ket{t},\ket{00},\ket{11}$. The eigenvalues are
\begin{equation}
    H_{\ell} \ket{\Psi}_{\ell} = \left( \frac{4 A^2 J z (3 M-2 z)}{n} \right) \ket{\Psi}_{\ell}
\end{equation}
Again, we see that we should set $A = \sqrt{K}$ since the gap to the first excited state (for $M \gg z$) goes like $\Delta = 12 J A^2 M / n$. So if we chose $A = 1$ the gap would scale like $M/n$, whereas if we choose $A = \sqrt{K}$ the gap is $\mathcal{O}(1)$.

Finally, we can use these results to calculate the expected overlap $\mathbb{E}[q_U(\mathbf{x}) A_U(\mathbf{x})]$:
\begin{align}
    \mathbb{E}[q_U(\mathbf{x}) A_U(\mathbf{x})] &= \bra{\mathbf{x\overline{x} x\overline{x}}} e^{-H_{\mathrm{eff}} T} \ket{\mathbf{0101}} \nonumber \\
    &= \bra{\mathbf{x\overline{x}}} e^{-H_0 T} \ket{\mathbf{01}} \prod_{\ell} \bra{\mathbf{x\overline{x}}} e^{-H_{\ell} T} \ket{\mathbf{01}}_{\ell} \nonumber \\
    &= \mathbb{E}[q_U(\mathbf{x})] \prod_{\ell} \mathbb{E}[q_{\ell}(\mathbf{x})]
\end{align}
so the overlap expectation value completely factorizes. Plugging in our known results, we find:
\begin{align}
    \mathbb{E}[q_U(\mathbf{x}) A_U(\mathbf{x})] &= \left[ \frac{1}{2^n} \sum_{\mathbf{z} = \{0,1\}} (-1)^{\mathbf{x} \cdot \mathbf{z}} \exp \left(- \frac{4 J T \magn{\mathbf{z}} (3 n - 2 \magn{\mathbf{z}} )}{n} \right) \right] \nonumber \\
    &\times \prod_{\ell=1}^K \left[ \frac{1}{2^M} \sum_{\mathbf{z}_{\ell} = \{0,1\}} (-1)^{\mathbf{x}_{\ell} \cdot \mathbf{z}_{\ell}} \exp \left(- \frac{4 A^2 J T \magn{\mathbf{z_{\ell}}} (3 M - 2 \magn{\mathbf{z_{\ell}}})}{n} \right) \right]
    \label{eq:pqoverlap}
\end{align}
where $\mathbf{x}_{\ell}$ is the length-$M$ substring of $\mathbf{x}$ that lives in subsystem $\ell$.
At long times $T \rightarrow \infty$, this expression limits to $\mathbb{E}[q_U(\mathbf{x}) A_U(\mathbf{x})] \rightarrow 2^{-2N}$.

The above expression simplifies if we consider the approximation $n \gg \magn{\mathbf{z}}, M \gg \magn{\mathbf{z}_{\ell}}$ and set $A = \sqrt{K}$, which gives us
\begin{align}
     \mathbb{E}[q_U(\mathbf{x}) A_U(\mathbf{x})]  &= \frac{1}{2^{2N}}\left[ \left(1 - e^{-12 J T} \right)^{\magn{\mathbf{x}}} \left(1 + e^{-12 J T} \right)^{n-\magn{\mathbf{x}}} \right] \nonumber \\
     & \quad \quad \times \prod_{\ell=1}^K \left[ \left(1 - e^{-12 J T} \right)^{\magn{\mathbf{x}_{\ell}}} \left(1 + e^{-12 J T} \right)^{M-\magn{\mathbf{x}_{\ell}}} \right] \nonumber \\
     &= \frac{1}{2^{2N}} \left[ \left(1 - e^{-12 J T} \right)^{2\magn{\mathbf{x}}} \left(1 + e^{-12 J T} \right)^{2(n-\magn{\mathbf{x}})} \right]
     \label{eq:harvardoverlapk2bitstring}
\end{align}
This is the expression we would find if we solved the model using large-$n$ methods. If $\mathbf{x} = \mathbf{0}$, then the above expression simplifies to
\begin{equation}
    \mathbb{E}[q_U(\mathbf{0}) A_U(\mathbf{0})] = \frac{1}{2^{2N}} \left(1 + e^{-12 J T} \right)^{2N}
\end{equation}
Finally, if we perform a sum over bitstrings, we obtain:
\begin{align}
    \sum_{\mathbf{x}} \mathbb{E}[q_U(\mathbf{x}) A_U(\mathbf{x})] &= \frac{1}{2^{2N}} \left[ \left(1 - e^{-12 J T} \right)^{2} + \left(1 + e^{-12 J T} \right)^{2} \right]^n \nonumber \\
    &= \frac{1}{2^{n}} \left( 1 + e^{-24 J T} \right)^n
    \label{eq:pqoverlapsumbits}
\end{align}
The only difference between this and Eq. \eqref{eq:k2sumbitstrings} is the leading factor of $2$. Therefore, the spoofing quantum circuit of \cite{gao2024limitations} gives an \textsf{XEB} score of
\begin{equation}
    \mathrm{XEB}(U,A) \approx n e^{-24 J T} + \mathcal{O}(e^{-48 J T})
\end{equation}

\subsubsection{Moments of Disjoint Circuit}

It is also useful to compute the moments $\mathbb{E}[A_U^k(\mathbf{x})]$ of the disjoint circuit. Because the disjoint subsystems are independent, we immediately have:
\begin{equation}
    \mathbb{E}[A_U^k(\mathbf{x})] = \prod_{\ell} \mathbb{E}[q_{\ell}^k(\mathbf{x})]
\end{equation}
where each $q_{\ell}(\mathbf{x})$ is the probability distribution for the $\ell$th subsystem. We have already calculated the $k$th moment for the full circuit $\mathbb{E}[q_U^k(\mathbf{x})]$. The large-$n$ result for this appears in Eq. \eqref{eq:arbitrarykfinal}, which we may simply reuse here with the substitutions $n \rightarrow M$ and $\mathbf{x} \rightarrow \mathbf{x_{\ell}}$. This gives:
\begin{align}
    \mathbb{E}[A_U^k(\mathbf{x})] &= \prod_{\ell = 1}^K \frac{k!}{2^{kM}} \left(1 - e^{-12 J T} \right)^{k \magn{\mathbf{x}_{\ell}}} \left(1 + e^{-12 J T} \right)^{k (M-\magn{\mathbf{x}_{\ell}})}  \nonumber \\
    &= \frac{(k!)^K}{2^{kN}} \left(1 - e^{-12 J T} \right)^{k \magn{\mathbf{x}}} \left(1 + e^{-12 J T} \right)^{k (n-\magn{\mathbf{x}})}
\end{align}
and performing the sum over bitstrings gives:
\begin{equation}
    \sum_{\mathbf{x}} \mathbb{E}[A_U^k(\mathbf{x})] = \frac{(k!)^K}{2^{kN}} \left[ \left(1 - e^{-12 J T} \right)^{k} +  \left(1 + e^{-12 J T} \right)^{k} \right]^n
\end{equation}

\subsubsection{Variance of algorithm of \cite{gao2024limitations}}

In the previous subsection we calculated the expected \textsf{XEB} score for the spoofing algorithm of \cite{gao2024limitations}, but we are also interested in how typical this score is. We can get a handle on this question by investigating the variance of the score:
\begin{align}
    \mathrm{Var \ XEB}(U,A) &= \langle (q_U A_U)^2 \rangle - \langle q_U A_U \rangle^2 \nonumber \\
    &= \frac{1}{2^n} \sum_{\mathbf{x}}\mathbb{E}\left[ \left(q_U(\mathbf{x}) A_U(\mathbf{x}) \right)^2 \right] - \left( \frac{1}{2^n} \sum_{\mathbf{x}} \mathbb{E}\left[ q_U(\mathbf{x}) A_U(\mathbf{x}) \right] \right)^2 
\end{align}
This involves calculating a fourth moment quantity with $k = 4$:
\begin{equation}
    \mathbb{E}\left[ q_U^2(\mathbf{x}) A_U^2(\mathbf{x}) \right] = \bra{\mathbf{x\overline{x} x\overline{x}} \ldots} \prod_t \mathbb{E}[ U_t \otimes U_t^{\mathcal{T}} \otimes U_t \otimes U_t^{\mathcal{T}} \otimes V_t \otimes V_t^{\mathcal{T}} \otimes V_t \otimes V_t^{\mathcal{T}}] \ket{\mathbf{0101} \ldots}
\end{equation}
It is not too much additional work to consider arbitrary moments of this quantity:
\begin{equation}
    \mathbb{E} [ q_U^c(\mathbf{x}) A_U^c(\mathbf{x}) ]
\end{equation}
where $k = 2c$ and $c \geq 1$. Taking $c = 1$ will be a good check on our results for the higher moments. Performing the disorder average over the couplings $J_{ij}^{\alpha \beta}(t)$ as above, we find an effective Hamiltonian
\begin{align}
    H_{\mathrm{eff}} &= \frac{J}{2 n} \sum_{\ell < m} \sum_{\substack{i \in I_{\ell} \\ j \in I_{m}}} \sum_{\alpha \beta} \left( \sum_{r = 1}^c \sum_{a=L,R} (-1)^a \ \sigma_{ira}^{\alpha} \sigma_{jra}^{\beta} \right)^2  \nonumber \\
    &+ \frac{J}{2N} \sum_{\ell} \sum_{i < j \in I_{\ell}} \sum_{\alpha \beta} \left( \sum_{r=1}^{2c} \sum_{a=L,R} (-1)^a A_r \ \sigma_{ira}^{\alpha} \sigma_{jra}^{\beta} \right)^2
    \label{eq:harvardalghamarbk}
\end{align}
where
\begin{equation}
    A_r = \begin{cases} 
      1 & r \leq c \\
      \sqrt{K} & r > c
   \end{cases}
\end{equation}
which generalizes Eq. \eqref{eq:harvardalgham}. This Hamiltonian has $(c!)^{K+1}$ ground states. We have the usual `ladder' ground state
\begin{equation}
    \ket{\Omega} = \bigotimes_{ir} \ket{s}_{irL,irR}
\end{equation}
with $H_{\mathrm{eff}} \ket{\Omega} = 0$ since we can pass Pauli operators through the singlets such that terms of opposite sign cancel in Eq. \ref{eq:harvardalghamarbk}. We can get $c!-1$ additional ground states by permuting replicas $r = 1,2,\ldots,c$ amongst themselves. Furthermore, for each $\ell$, we can permute replicas $r = c+1,c+2,\ldots,2c$ amongst themselves, which causes the same cancellations. Crucially, we can do this second permutation independently for each subsystem $\ell$, so we have a total of $(c!)^{1+K}$ independent ground states.

After some manipulation, and after ignoring terms that are subleading in $n$, the effective Hamiltonian can be written as
\begin{equation}
    H_{\mathrm{eff}} = H_0^{(c)} + \sum_{\ell} H_{\ell}^{(c)}
    \label{eq:harvardalgeffhamarbc}
\end{equation}
where
\begin{equation}
    H_0^{(c)} = \frac{J}{2N} \sum_{ra<sb}^{r,s \leq c} (-1)^{a+b} \left( \sum_{i=1}^n \heis{ira}{isb} \right)^2 + \frac{9 J}{2} c n
\end{equation}
and
\begin{equation}
    H_{\ell}^{(c)} = \frac{J}{2N} K \sum_{ra < sb}^{c < r,s \leq 2c} (-1)^{a+b} \left( \sum_{i \in I_{\ell}} \heis{ira}{isb} \right)^2 + \frac{9 J}{2} c M
\end{equation}
which generalizes Eq. \eqref{eq:harvardh0hl}. As a quick check we can set $\heis{irL}{irR} = -3$ and all others 0 to find $H_0^{(c)} = 0$ and $H_{\ell}^{(c)} = 0$.

Finally, we can use these results to calculate the $c$th moment of the expected overlap $\mathbb{E}[q_U^c(\mathbf{x}) A_U^c(\mathbf{x})]$:
\begin{align}
    \mathbb{E}[q_U^c(\mathbf{x}) A_U^c(\mathbf{x})] &= \bra{\mathbf{x\overline{x} x\overline{x}} \cdots} e^{-H_{\mathrm{eff}} T} \ket{\mathbf{0101} \cdots} \nonumber \\
    &= \bra{\mathbf{x\overline{x}} \cdots} e^{-H_0^{(c)} T} \ket{\mathbf{01} \cdots} \prod_{\ell} \bra{\mathbf{x\overline{x}} \cdots} e^{-H_{\ell}^{(c)} T} \ket{\mathbf{01} \cdots}_{\ell} \nonumber \\
    &= \mathbb{E}[q_U^c(\mathbf{x})] \prod_{\ell} \mathbb{E}[A^{c}_{U,\ell}(\mathbf{x})]
\end{align}
so the overlap expectation value completely factorizes as before. We can use our previous results (see Eq. \eqref{eq:arbitrarykfinal}) to obtain an explicit expression for this quantity. We have:
\begin{align}
     \mathbb{E}[q_U^c(\mathbf{x}) A_U^c(\mathbf{x})]  &= \left[ \frac{c!}{2^{cN}} \left(1 - e^{-12 J T} \right)^{c\magn{\mathbf{x}}} \left(1 + e^{-12 J T} \right)^{c(n-\magn{\mathbf{x}})} \right] \nonumber \\
     & \quad \quad \times \prod_{\ell=1}^K \left[ \frac{c!}{2^{cM}} \left(1 - e^{-12 J T} \right)^{c\magn{\mathbf{x}_{\ell}}} \left(1 + e^{-12 J T} \right)^{c(M-\magn{\mathbf{x}_{\ell}})} \right] \nonumber \\
     &= \frac{(c!)^{1+K}}{2^{2cN}} \left[ \left(1 - e^{-12 J T} \right)^{2c\magn{\mathbf{x}}} \left(1 + e^{-12 J T} \right)^{2c(n-\magn{\mathbf{x}})} \right]
     \label{eq:pqcthmoment}
\end{align}
We see that this reduces to Eq. \eqref{eq:harvardoverlapk2bitstring} for $c = 1$. Finally, performing a sum over bitstrings, we obtain:
\begin{equation}
    \sum_{\mathbf{x}} \mathbb{E}[q_U^c(\mathbf{x}) A_U^c(\mathbf{x})] = \frac{(c!)^{1+K}}{2^{2cN}} \left[ \left(1 - e^{-12 J T} \right)^{2c} + \left(1 + e^{-12 J T} \right)^{2c} \right]^n
    \label{eq:pqcthmomentsumbitstrings}
\end{equation}
We see that this reduces to Eq. \eqref{eq:pqoverlapsumbits} for $c = 1$.

We can use these results to compute the variance of the spoofing algorithm of \cite{gao2024limitations}. For $K = n / \log n$, the prefactor in Eq. \eqref{eq:pqcthmomentsumbitstrings} completely dominates the square of the mean, so we have
\begin{equation}
    \mathrm{Var \ XEB}(U,A) \rightarrow \frac{2^{1+K} - 1}{2^{4N}} \approx \frac{2^{1+K}}{2^{4N}}
    \label{eq:varspoof}
\end{equation}
at long times $T \rightarrow \infty$, whereas the variance of the true distribution goes as
\begin{equation}
    \mathrm{Var \ XEB}(U,U) \rightarrow \frac{4! - 4}{2^{4N}} = \frac{20}{2^{4N}}
    \label{eq:vartrue}
\end{equation}

\subsubsection{algorithm of \cite{gao2024limitations} with Strong 1-Design Property}

The version of the algorithm of \cite{gao2024limitations} for Brownian circuits presented above has a 1-design time that scales similarly to the 2-design time. This makes it somewhat different from discrete random circuits, which become 1-designs after a single layer of gates. To bring our model closer to these discrete circuits, we consider a variant of the algorithm with a shorter 1-design time generated by strong single-qubit scrambling. To generate this single-qubit scrambling, we consider adding the following Brownian dynamics in between each two-qubit scrambling unitary:
\begin{equation}
    W_t = \exp{\left( -i \sum_{i \alpha} \mu_i^{\alpha}(t) \sigma_i^{\alpha} \Delta t \right)}
\end{equation}
where the coupling coefficients $\mu_i^{\alpha}(t)$ are Gaussian white-noise variables with zero mean and variance
\begin{equation}
    \mathbb{E} \left[ \mu_i^{\alpha}(t) \mu_{i'}^{\alpha'}(t') \right] = \delta_{ii'} \delta^{\alpha \alpha'} \delta_{tt'} \frac{\mu}{\Delta t}
\end{equation}
Notice that there is no factor of $n$ in the denominator of the variance; this guarantees that the resulting effective Hamiltonian is extensive. These single-qubit rotations are the same across all replicas, including the disjoint circuit, so the expectation value, to lowest order in $\Delta t$, is
\begin{align}
    \mathbb{E} \left[ W_t \otimes W_t^{\mathcal{T}} \otimes W_t \otimes W_t^{\mathcal{T}} \right] &= \prod_{i \alpha} \left( \int d \mu_{i}^{\alpha}(t) \exp\left[ - \frac{\left( \mu_{i}^{\alpha}(t) \right)^2 }{2\mu/\Delta t n} \right] \exp \left[ - i \mu_{i}^{\alpha}(t) \Delta t \left( \sum_{ra} \sigma_{ira}^{\alpha}  \right) \right] \right) \nonumber \\
    &= \exp \left[ - \frac{\mu}{2} \sum_{i \alpha} \left( \sum_{ra} \sigma_{ira}^{\alpha} \right)^2 \Delta t \right].
\end{align}
Notice that there are no minus signs with respect to $a = L,R$, because $W_t^{\mathcal{T}} = W_t$. Stacking these terms up in time, we therefore find an effective Hamiltonian
\begin{equation}
    H^{(1)}_{\mathrm{eff}} = \frac{\mu}{2} \sum_{i \alpha} \left( \sum_{ra} \sigma_{ira}^{\alpha} \right)^2
\end{equation}
which, after some manipulation, can be written as
\begin{align}
    H^{(1)}_{\mathrm{eff}} &= 3 \mu n k + \mu \sum_i \sum_{ra < sb} \vec{\sigma}_{ira} \cdot \vec{\sigma}_{isb} \nonumber \\
    &= 2 \mu \sum_i \vec{S}_i \cdot \vec{S}_i
\end{align}
where $S^{\alpha}_i = \frac{1}{2} \sum_{ra} \sigma^{\alpha}_i$ is the total spin on site $i$. For large $\mu$, this Hamiltonian therefore wants to minimize the total spin, i.e. $\vec{S}_i \cdot \vec{S}_i = 0$.
This effective Hamiltonian is to be added to the 2-qubit scrambling Hamiltonian Eq.~\eqref{eq:harvardeffectivehamtotal}. Ignoring terms that are subleading in $n$, we therefore have the total effective Hamiltonian
\begin{equation}
    H_{\mathrm{eff}} = H_N + H_{AM} + H^{(1)}_{\mathrm{eff}} + \mathcal{O}(A^2)
\end{equation}
We shall keep the $\mathcal{O}(AM)$ term for now, but we shall see in a moment that it is negligible in the saddle-point calculation. We shall also take $\mu \gg J$ such that the 1-design time is much shorter than the 2-design time. This effective Hamiltonian has a single unique ground state given by the familiar `ladder' diagram. Note that this state consists of pairs of spin singlets, so that we have $\vec{S}_i \cdot \vec{S}_i = 0$ as enforced by the new term $H_{\mathrm{eff}}^{(1)}$.

Converting this effective Hamiltonian to a path integral and introducing the mean fields 
\begin{equation}
    G_{\ell}^{rasb} = \frac{1}{M} \sum_{i \in I_{\ell}} \vec{\sigma}_{ira} \cdot \vec{\sigma}_{isb}
\end{equation}
and corresponding Lagrange multipliers $F_{\ell}^{rasb}$, we find the following action (assuming the saddle points are time-independent):
\begin{align}
    I[F_{\ell}^{rasb},G_{\ell}^{rasb}] &= - \sum_{\ell} \sum_{i \in I_{\ell}} \ln \left( \left[(-1)^{x_i}\right]^k \bra{x_i \overline{x}_i \cdots} \exp \left[ - \sum_{ra < sb} F_{\ell}^{rasb} \vec{\sigma}_{ira} \cdot \vec{\sigma}_{isb} T \right] \ket{01 \cdots} \right) \nonumber \\
    &- \frac{J M^2}{2N} \sum_{\ell m} G_{\ell}^{1L1R} G_{m}^{1L1R} T - \frac{J A^2 M^2}{2N} \sum_{\ell} \left( G_{\ell}^{2L2R} \right)^2 T \nonumber \\
    &+ \frac{J A M^2}{2N} \sum_{\ell} \left[ \left( G_{\ell}^{1L2L} \right)^2 + \left( G_{\ell}^{1R2R} \right)^2 - \left( G_{\ell}^{1L2R} \right)^2 - \left( G_{\ell}^{2L1R} \right)^2 \right] T \nonumber \\
    &+ \mu M \sum_{\ell} \sum_{ra < sb} G_{\ell}^{rasb} T - M \sum_{\ell} \sum_{ra < sb} F_{\ell}^{rasb} G_{\ell}^{rasb} T \nonumber \\
    &+ \frac{JT}{2N} \left[ 9 n^2 + 9 A^2 K M^2 \right] + 6 n \mu T
\end{align}
The saddle-point equation of motion $\partial I / \partial G_{\ell}^{rasb} = 0$ for the mean fields gives
\begin{align}
    F_{\ell}^{1L1R} &= \mu - \frac{J M}{n} \sum_{m} G_m^{1L1R} \nonumber \\
    F_{\ell}^{2L2R} &= \mu - \frac{J A^2 M}{n} G_{\ell}^{2L2R} \nonumber \\
    F_{\ell}^{rasb} &= \mu + (-1)^{a+b} \frac{J A M}{n} G_{\ell}^{rasb}, \quad \quad \quad \quad rasb = 1L2L, 1R2R, 1L2R, 2L1R \nonumber \\
    &\approx \mu
    \label{eq:Fellrasbstrong1design}
\end{align}
where in the last line we have dropped the doubly-suppressed term because $J \ll \mu$ and $AM \ll n$. Plugging these equations of motion back into the action, we find:
\begin{align}
    I &= - \sum_{\ell} \sum_{i \in I_{\ell}} \ln \left( \left[(-1)^{x_i}\right]^k \bra{x_i \overline{x}_i \cdots} \exp \left[ - \sum_{ra < sb} F_{\ell}^{rasb} \vec{\sigma}_{ira} \cdot \vec{\sigma}_{isb} T \right] \ket{01 \cdots} \right) \nonumber \\
    &+ \frac{J M^2}{2N} \sum_{\ell m} G_{\ell}^{1L1R} G_{m}^{1L1R} T + \frac{J A^2 M^2}{2N} \sum_{\ell} \left( G_{\ell}^{2L2R} \right)^2 T \nonumber \\
    &- \frac{J A M^2}{2N} \sum_{\ell} \left[ \left( G_{\ell}^{1L2L} \right)^2 + \left( G_{\ell}^{1R2R} \right)^2 - \left( G_{\ell}^{1L2R} \right)^2 - \left( G_{\ell}^{2L1R} \right)^2 \right] T \nonumber \\
    &+ \frac{JT}{2N} \left[ 9 n^2 + 9 A^2 K M^2 \right] + 6 n \mu T
\end{align}

Our next task is to solve for the spectrum of the Hamiltonian
\begin{equation}
    H = \sum_{ra < sb} F_{\ell}^{rasb} \vec{\sigma}_{ra} \cdot \vec{\sigma}_{sb}
\end{equation}
where the coefficients $F_{\ell}^{rasb}$ are given by Eq.~\eqref{eq:Fellrasbstrong1design} and we have dropped the index $i$ for the moment. Fortunately, this Hamiltonian is exactly solvable. To see this, take $a_1 = - \frac{J M}{n} \sum_{m} G_m^{1L1R}$ and $a_2 = - \frac{J A^2 M}{n} G_{\ell}^{2L2R}$ so that we may write
\begin{equation}
    H = \mu \sum_{ra < sb} \vec{\sigma}_{ra} \cdot \vec{\sigma}_{sb} + \sum_r a_r \ \vec{\sigma}_{rL} \cdot \vec{\sigma}_{rR}.
    \label{eq:musinglesiteham}
\end{equation}
Then define the spins
\begin{align}
    \vec{L}_r &= \frac{1}{2} \sum_a \vec{\sigma}_{ra} \nonumber \\
    \vec{S} &= \sum_r \vec{L}_r = \frac{1}{2} \sum_{ra} \vec{\sigma}_{ra} 
\end{align}
and notice that the Hamiltonian can be written entirely in terms of the squares of these spins:
\begin{equation}
    H = 2 \mu \vec{S} \cdot \vec{S} - 3 k \mu + \sum_r a_r \left( 2 \vec{L}_r \cdot \vec{L}_r - 3 \right).
    \label{eq:hamsquarespins}
\end{equation}
These squares all mutually commute:
\begin{align}
    [\vec{L}_r \cdot \vec{L}_r, \vec{L}_s \cdot \vec{L}_s] &= 0 \nonumber \\
    [\vec{L}_r \cdot \vec{L}_r, \vec{S} \cdot \vec{S}] &= 0
\end{align}
so the Hamiltonian spectrum is completely solved. When the single-qubit couplings are much stronger than the two-qubit couplings $\mu \gg J$, the lowest eigenstates of this Hamiltonian have $\vec{S} \cdot \vec{S} = 0$. There are two ways this can happen: either we have $\vec{L}_1 \cdot \vec{L}_1 = \vec{L}_2 \cdot \vec{L}_2 = 0$ (i.e. spin-0 singlets for each replica $r = 1,2$) or $\vec{L}_1 \cdot \vec{L}_1 = \vec{L}_2 \cdot \vec{L}_2 = 2$ (i.e. a pair of spin-1 particles that combine into a total spin-0). These two possibilities give the following energy eigenvalues:
\begin{align}
    E_0 &= -6 \mu - 3 a_1 - 3 a_2 \nonumber \\
    E_1 &= -6 \mu + a_1 + a_2
\end{align}
We therefore have the low-energy spectrum of $H$ as
\begin{equation}
    H = E_0 \ket{\psi_0} \bra{\psi_0} + E_1 \ket{\psi_1} \bra{\psi_1} + \cdots
\end{equation}
where
\begin{align}
    \ket{\psi_0} &= \frac{1}{2} \left( \ket{0101} + \ket{1010} - \ket{0110} - \ket{1001} \right) \nonumber \\
    \ket{\psi_1} &= \frac{1}{2 \sqrt{3}} \left( 2 \ket{0011} + 2 \ket{1100} - \ket{0101} - \ket{1010} - \ket{0110} - \ket{1001} \right).
\end{align}
The fact that the $1R1L$ replicas are entangled with the $2R2L$ replicas in $\ket{\psi_1}$ implies that the distribution $\mathbb{E}[q_U(\mathbf{x}) A_U(\mathbf{x})] $ is not separable, as was the case earlier.

We now plug this result back into the action and find
\begin{align}
    I &= - M \sum_{\ell} \ln \left( \frac{1}{4} e^{-E_{0,\ell} T} + \frac{1}{12} e^{-E_{1,\ell} T} + \cdots \right) \nonumber \\
    &+ \frac{J M^2}{2N} \sum_{\ell m} G_{\ell}^{1L1R} G_{m}^{1L1R} T + \frac{J A^2 M^2}{2N} \sum_{\ell} \left( G_{\ell}^{2L2R} \right)^2 T \nonumber \\
    &+ \frac{JT}{2N} \left[ 9 n^2 + 9 A^2 K M^2 \right] + 6 n \mu T
\end{align}
where the equations of motion set $G^{1L2L} = G^{1R2R} = G^{1L2R} = G^{2L1R} = 0$, and we have used the overlaps
\begin{align}
    \bracket{0101}{\psi_0} &= \bracket{1010}{\psi_0} = \frac{1}{2} \nonumber \\
    \bracket{0101}{\psi_1} &= \bracket{1010}{\psi_1} = -\frac{1}{2 \sqrt{3}} 
\end{align}
and the energies are
\begin{align}
    E_{0,\ell} &= -6 \mu + 3 \frac{J M}{n} \sum_{m} G_m^{1L1R} + 3 \frac{J A^2 M}{n} G_{\ell}^{2L2R} \nonumber \\
    E_{1,\ell} &= -6 \mu - \frac{J M}{n} \sum_{m} G_m^{1L1R} - \frac{J A^2 M}{n} G_{\ell}^{2L2R}.
\end{align}
Taking derivatives $\partial I / \partial G_{\ell}^{rLrR}$, we find the equations of motion
\begin{equation}
    G_{\ell}^{1L1R} = G_{\ell}^{2L2R} = -3 + \cdots
\end{equation}
where $\cdots$ indicates time-dependent terms falling off as $e^{-(\#) J T}$. Plugging these solutions back into the action and taking $A = \sqrt{K}$, we finally find
\begin{align}
    \mathbb{E}[q_U(\mathbf{x}) A_U(\mathbf{x})] = e^{- I^*} &= \frac{1}{2^{2N}} \left( 1 + \frac{1}{3} e^{-24 J T} + \cdots \right)^n \nonumber \\
    &\approx \frac{1}{2^{2N}} \left( 1 + \frac{1}{3} n e^{-24 J T} + \cdots \right)
    \label{eq:pq1designexpval}
\end{align}
Notice that this result is independent of the bitstring $\mathbf{x}$, which makes sense, since we are strongly scrambling each qubit. We therefore find an \textsf{XEB} score of
\begin{equation}
    \mathrm{XEB} \approx \frac{1}{3} n e^{-24 JT}
\end{equation}

\subsubsection{Variance of algorithm of \cite{gao2024limitations} with Strong 1-Design Property}

We may calculate the variance $\mathbb{E}[q_U^2(\mathbf{x}) A_U^2(\mathbf{x})]$ using similar techniques. Since the derivation is not much more difficult for arbitrary moments, we shall calculate the moment $\mathbb{E}[q_U^c(\mathbf{x}) A_U^c(\mathbf{x})]$ for any integer $c \geq 1$. Anticipating that the $\mathcal{O}(AM)$ terms will again be negligible, we take our effective Hamiltonian to be, starting from Eq.~\eqref{eq:harvardalgeffhamarbc},
\begin{equation}
    H_{\mathrm{eff}} = H_0^{(c)} + \sum_{\ell} H_{\ell}^{(c)} + H_{\mathrm{eff}}^{(1)}
\end{equation}
Converting this into a path integral and introducing the fields $F_{\ell}^{rasb},G_{\ell}^{rasb}$, we find the action
\begin{align}
    I[F_{\ell}^{rasb},G_{\ell}^{rasb}] &= - \sum_{\ell} \sum_{i \in I_{\ell}} \ln \left( \bra{x_i \overline{x}_i \cdots} \exp \left[ - \sum_{ra < sb} F_{\ell}^{rasb} \vec{\sigma}_{ira} \cdot \vec{\sigma}_{isb} T \right] \ket{01 \cdots} \right) \nonumber \\
    &+ \frac{J M^2}{2N} \sum_{\ell m} \sum_{ra < sb}^{r,s \leq c} (-1)^{a+b} G_{\ell}^{rasb} G_{m}^{rasb} T + \frac{J A^2 M^2}{2N} \sum_{\ell} \sum_{ra < sb}^{c < r,s \leq 2c} (-1)^{a+b} \left( G_{\ell}^{rasb} \right)^2 T \nonumber \\
    &+ \mu M \sum_{\ell} \sum_{ra < sb} G_{\ell}^{rasb} T - M \sum_{\ell} \sum_{ra < sb} F_{\ell}^{rasb} G_{\ell}^{rasb} T \nonumber \\
    &+ \frac{JT}{2N} \left[ 9 n^2 c + 9 A^2 K M^2 c \right] + 6 n \mu T c 
    \label{eq:harvardvar1designaction}
\end{align}
The saddle-point equations of motion $\partial I / \partial G_{\ell}^{rasb} = 0$ for the mean fields gives
\begin{equation}
    F_{\ell}^{rasb} = \begin{cases}
        \mu + (-1)^{a+b} \frac{J M}{n} \sum_{m} G_m^{rasb} & r,s \leq c \\
        \mu + (-1)^{a+b} \frac{J A^2 M}{n} G_{\ell}^{rasb} & c < r,s \leq 2c \\
        \mu & \mathrm{otherwise}
    \end{cases}
\end{equation}
Plugging these solutions back into the action we obtain
\begin{align}
    I &= - \sum_{\ell} \sum_{i \in I_{\ell}} \ln \left( \bra{x_i \overline{x}_i \cdots} \exp \left[ - \sum_{ra < sb} F_{\ell}^{rasb} \vec{\sigma}_{ira} \cdot \vec{\sigma}_{isb} T \right] \ket{01 \cdots} \right) \nonumber \\
    &- \frac{J M^2}{2N} \sum_{\ell m} \sum_{ra < sb}^{r,s \leq c} (-1)^{a+b} G_{\ell}^{rasb} G_{m}^{rasb} T - \frac{J A^2 M^2}{2N} \sum_{\ell} \sum_{ra < sb}^{c < r,s \leq 2c} (-1)^{a+b} \left( G_{\ell}^{rasb} \right)^2 T \nonumber \\
    &+ \frac{JT}{2N} \left[ 9 n^2 c + 9 A^2 K M^2 c \right] + 6 n \mu T c
\end{align}
Again, inspired by the expected ground states, we make the ansatz that $G_{\ell}^{rasb} = 0$ unless $r = s$ and $ab = LR$, which should correspond to the familiar `ladder' diagram. Making this ansatz, the action reduces to
\begin{align}
    I &= - \sum_{\ell} \sum_{i \in I_{\ell}} \ln \left( \bra{x_i \overline{x}_i \cdots} \exp \left[ - \sum_{ra < sb} F_{\ell}^{rasb} \vec{\sigma}_{ira} \cdot \vec{\sigma}_{isb} T \right] \ket{01 \cdots} \right) \nonumber \\
    &+ \frac{J M^2}{2N} \sum_{\ell m} \sum_{r \leq c} G_{\ell}^{rLrR} G_{m}^{rLrR} T + \frac{J A^2 M^2}{2N} \sum_{\ell} \sum_{c < r \leq 2c} \left( G_{\ell}^{rLrR} \right)^2 T \nonumber \\
    &+ \frac{JT}{2N} \left[ 9 n^2 c + 9 A^2 K M^2 c \right] + 6 n \mu T c
\end{align}
with Lagrange multipliers
\begin{equation}
    F_{\ell}^{rLrR} = \begin{cases}
        \mu - \frac{J M}{n} \sum_{m} G_m^{rLrR} & r \leq c \\
        \mu - \frac{J A^2 M}{n} G_{\ell}^{rLrR} & c < r \leq 2c \\
    \end{cases}
\end{equation}
with all others $F_{\ell}^{rasb} = \mu$. We first show that this ansatz leads to a self-consistent saddle-point before returning to discuss other possible saddle-points.

In this case, we must again solve for the spectrum of Eq.~\eqref{eq:musinglesiteham}, but now $r$ ranges from $r = 1,\ldots,2c$, and
\begin{equation}
    a_r = \begin{cases} 
      - \frac{J M}{n} \sum_{m} G_m^{rLrR} & r \leq c \\
      - \frac{J A^2 M}{n} G_{\ell}^{rLrR} & c < r \leq 2c
   \end{cases}
\end{equation}
As above, we can rewrite this Hamiltonian using total spins $\vec{L}_r, \vec{S}$ to obtain Eq.~\eqref{eq:hamsquarespins}, where $k = 2c$. When $\mu \gg J$, the lowest energy eigenstate $\ket{\psi_0}$ is a collection of spin singlets, where $\vec{L}_r \cdot \vec{L}_r = \vec{S} \cdot \vec{S} = 0$ and
\begin{equation}
    E_0 = -6 c \mu - 3 \sum_r a_r
\end{equation}
The overlap with the boundary states is
\begin{equation}
    \bracket{0101 \cdots}{\psi_0} = \bracket{1010 \cdots}{\psi_0} = \frac{1}{2^c}.
\end{equation}
The next excited states have pairs of spin-1 particles on replicas $u,v$: $\vec{L}_u \cdot \vec{L}_u = \vec{L}_v \cdot \vec{L}_v = 2$, and singlets on all other replicas $\vec{L}_r \cdot \vec{L}_r = 0 \ \forall \ r \neq u,v$. There are a total of $c(2c - 1)$ of these excited states $\ket{\psi_{u,v}}$, with energies
\begin{equation}
    E_{u,v} = -6 c \mu + a_u + a_v - 3 \sum_{r \neq u,v} a_r
\end{equation}
for $1 \leq u < v \leq 2c$. The overlap with the boundary states is
\begin{equation}
    \bracket{0101 \cdots}{\psi_{u,v}} = \bracket{1010 \cdots}{\psi_{u,v}} = \frac{1}{2^{c-1}} \cdot \frac{1}{2 \sqrt{3}}.
\end{equation}
We therefore have the low-energy spectrum of $H$ as
\begin{equation}
    H = E_0 \ket{\psi_0} \bra{\psi_0} + \sum_{u,v} E_{u,v} \ket{\psi_{u,v}} \bra{\psi_{u,v}} + \cdots
\end{equation}

We now plug this result back into the action to find
\begin{align}
    I &= - M \sum_{\ell} \ln \left( \frac{1}{2^{2c}} e^{-E_{0,\ell} T} + \sum_{u,v} \frac{1}{2^{2c}}\frac{1}{3} e^{-E_{u,v,\ell} T} + \cdots \right) \nonumber \\
    &+\frac{J M^2}{2N} \sum_{\ell m} \sum_{r \leq c} G_{\ell}^{rLrR} G_{m}^{rLrR} T + \frac{J A^2 M^2}{2N} \sum_{\ell} \sum_{c < r \leq 2c} \left( G_{\ell}^{rLrR} \right)^2 T \nonumber \\
    &+ \frac{J T}{2N} \left[ 9 n^2 c + 9 A^2 K M^2 c \right] + 6 n \mu T c
\end{align}
Taking derivatives $\partial I / \partial G_{\ell}^{rLrR}$ we find
\begin{equation}
    G_{\ell}^{rLrR} = -3 + \cdots \quad \quad \forall \ r
\end{equation}
where $\cdots$ represents time-dependent terms falling off as $e^{-(\#) J T}$. Plugging these solutions back into the action and taking $A = \sqrt{K}$, we find
\begin{align}
    I^* &= - n \ln \left( \frac{1}{2^{2c}} e^{6 c \mu T + 18 c J T} + \frac{1}{2^{2c}} \frac{c(2c-1)}{3} e^{6 c \mu T + 18 (c-1) J T - 6 J T} \right) \nonumber \\
    &+ \frac{J T}{n} \left[ 9 n^2 c + 9 A^2 K M^2 c \right] + 6 n \mu T c
\end{align}
and hence
\begin{equation}
    e^{-I^*} = \frac{1}{2^{2cN}} \left( 1 + \frac{c(2c-1)}{3} e^{-24 J T} + \cdots \right)^n
\end{equation}
which reproduces Eq.~\eqref{eq:pq1designexpval} for $c = 1$.

We are not yet finished, because we must consider other possible saddle points. Starting from the `ladder' saddle point, we can find other saddle points by permuting the $R$ replicas amongst themselves. For example, choosing $G_{\ell}^{rasb} = 0$ unless $s = \sigma(r)$ and $ab = LR$, where $\sigma$ is any permutation among the replicas $s = 1,\ldots,c$ leads to an action with the same form as the `ladder' action. Similarly, choosing $G_{\ell}^{rasb} = 0$ unless $s = \sigma_{\ell}(r)$ and $ab = LR$, where the $\sigma_{\ell}$ for each $\ell = 1,\ldots,K$ are independent permutations among the replicas $s = c+1,\ldots,2c$ also leads to the same action as the `ladder' action. To see this explicitly, let us choose a particular set of permutations $\sigma,\sigma_{\ell}$ and make the ansatz just described. In this case, the action reduces to
\begin{align}
    I &= - \sum_{\ell} \sum_{i \in I_{\ell}} \ln \left( \bra{x_i \overline{x}_i \cdots} \exp \left[ - \sum_{ra < sb} F_{\ell}^{rasb} \vec{\sigma}_{ira} \cdot \vec{\sigma}_{isb} T \right] \ket{01 \cdots} \right) \nonumber \\
    &+ \frac{J M^2}{2N} \sum_{\ell m} \sum_{r \leq c} G_{\ell}^{rL\sigma(r)R} G_{m}^{rL\sigma(r)R} T + \frac{J A^2 M^2}{2N} \sum_{\ell} \sum_{c < r \leq 2c} \left( G_{\ell}^{rL\sigma_{\ell}(r)R} \right)^2 T \nonumber \\
    &+ \frac{JT}{2N} \left[ 9 n^2 c + 9 A^2 K M^2 c \right] + 6 n \mu T c
\end{align}
with Lagrange multipliers
\begin{align}
    F_{\ell}^{rL\sigma(r)R} &=
        \mu - \frac{J M}{n} \sum_{m} G_m^{rL\sigma(r)R} \quad \quad \quad r \leq c \nonumber \\
    F_{\ell}^{rL\sigma_{\ell}(r)R} &=
        \mu - \frac{J A^2 M}{n} G_{\ell}^{rL\sigma_{\ell}(r)R} \quad \quad c < r \leq 2c \\
\end{align}
with all others $F_{\ell}^{rasb} = \mu$. Then we must solve for the spectrum of the Hamiltonian
\begin{equation}
    H_{\ell} = \mu \sum_{ra < sb} \vec{\sigma}_{ra} \cdot \vec{\sigma}_{sb} + \sum_{r \leq c} a_r \ \vec{\sigma}_{rL} \cdot \vec{\sigma}_{\sigma(r)R} + \sum_{c < r \leq 2c} a_r \vec{\sigma}_{rL} \cdot \vec{\sigma}_{\sigma_{\ell}(r)R}.
\end{equation}
for each $\ell$. Although the form of this Hamiltonian apears to be slightly more complicated (because we had to separate out the terms $r \leq c$ and $c < r \leq 2c$), it has exactly the same spectrum as the previous Hamiltonian. Moreover, the resulting action has the same form as previously, with $G_{\ell}^{rLrR} \rightarrow G_{\ell}^{rL\sigma(r)R}$ for $r \leq c$ and $G_{\ell}^{rLrR} \rightarrow G_{\ell}^{rL\sigma_{\ell}(r)R}$ for $c < r \leq 2c$. So we will get exactly the same action as previously.
Because the permutations $\sigma,\sigma_{\ell}$ are all independent of each other, we end up with $(c!)^{1+K}$ total saddle points, each with the same action. We therefore conclude
\begin{equation}
    \mathbb{E}[q_U^c(\mathbf{x}) A_U^c(\mathbf{x})] = (c!)^{1+K} e^{- I^*} = \frac{(c!)^{1+K}}{2^{2cN}} \left( 1 + \frac{c(2c-1)}{3} e^{-24 J T} + \cdots \right)^n
\end{equation}

\end{document}